\documentclass[10pt,reqno]{amsart}

\usepackage{amsfonts,amssymb,amsmath,amsopn,amsthm,graphicx}
\usepackage{amsxtra, mathrsfs}
\usepackage{amsbsy,dsfont}
\usepackage{colordvi}
\usepackage[usenames,dvipsnames]{color}
\usepackage{mathtools}
\usepackage[colorlinks=true]{hyperref}
\usepackage{enumitem}
\usepackage{bbm} % to get \1
\usepackage[parfill]{parskip}    % Activate to begin paragraphs with an empty line rather than an indent
\usepackage{lmodern} % to use \textsf in theorem environment
\newtheorem{theorem}{Theorem}[section]
\newtheorem{corollary}[theorem]{Corollary}
\newtheorem{lemma}[theorem]{Lemma}
\newtheorem{proposition}[theorem]{Proposition}

\theoremstyle{definition}

\newtheorem{remark}[theorem]{Remark}
\newtheorem{remarks}[theorem]{Remarks} 
\newtheorem{assumption}[theorem]{Assumption}

\numberwithin{equation}{section}

\newcommand{\1}{\mathbbm{1}}

\newcommand{\C}{\mathbb{C}}

\newcommand{\const}{\mathrm{const}\ }

\renewcommand{\epsilon}{\varepsilon}

\newcommand{\Lp}{\textsf{L}}

\newcommand{\h}{\mathcal{H}}

\newcommand{\loc}{{\rm loc}}

\renewcommand{\phi}{\varphi}
\newcommand{\rr}{\mathcal{R}}
\newcommand{\R}{\mathbb{R}}

\newcommand{\s}{\mathcal{S}}

\newcommand{\Z}{\mathbb{Z}}

\newcommand{\m}{{\scriptscriptstyle-}}
\newcommand{\pp}{{\scriptscriptstyle+}}
\newcommand{\ppm}{{\scriptscriptstyle\pm}}
\newcommand{\mpp}{{\scriptscriptstyle\mp}}

\DeclareMathOperator{\re}{Re}
\DeclareMathOperator{\spec}{spec}

\DeclareMathOperator{\sgn}{sgn}
\DeclareMathOperator{\tr}{Tr}

\DeclareMathOperator*{\essinf}{ess\,inf}
\DeclareMathOperator*{\esssup}{ess\,sup}

\begin{document}

\title[Lieb--Thirring inequality for the 2D Pauli operator]{Lieb--Thirring  inequality for the 2D Pauli operator}

\author{Rupert L. Frank}
\address[Rupert L. Frank]{Mathematisches Institut, Ludwig-Maximilians-Universit\"at M\"unchen, Theresienstr. 39, 80333 M\"unchen, Germany, and Mathematics 253-37, Caltech, Pasa\-de\-na, CA 91125, USA}
\email{r.frank@lmu.de, rlfrank@caltech.edu}

\author {Hynek Kova\v{r}\'{\i}k}
\address [Hynek Kova\v{r}\'{\i}k]{DICATAM, Sezione di Matematica, Universit\`a degli studi di Brescia, Italy}
\email {hynek.kovarik@unibs.it}

\thanks{\copyright\, 2024 by the authors. This paper may be reproduced, in its entirety, for non-commercial purposes.\\
	Partial support through US National Science Foundation grant DMS-1954995 (R.L.F.), as well as through the Deutsche Forschungsgemeinschaft (DFG, German Research Foundation) through Germany’s Excellence Strategy EXC-2111-390814868 (R.L.F.) is acknowledged.}

\begin{abstract}
	By the Aharonov--Casher theorem, the Pauli operator $P$ has no zero eigenvalue when the normalized magnetic flux $\alpha$ satisfies $|\alpha|<1$, but it does have a zero energy resonance. We prove that in this case a Lieb--Thirring inequality for the $\gamma$-th moment of the eigenvalues of $P+V$ is valid under the optimal restrictions $\gamma\geq |\alpha|$ and $\gamma>0$. Besides the usual semiclassical integral, the right side of our inequality involves an integral where the zero energy resonance state appears explicitly. Our inequality improves earlier works that were restricted to moments of order $\gamma\geq 1$.
\end{abstract}

\maketitle

\section{Introduction and main result}

\subsection{Background}
\label{ssec-background}
We are interested in quantitative information on the negative eigenvalues of the operator
$$
P + V
\qquad\text{in}\ \Lp^2(\R^2,\C^2) \,,
$$
where $P$ is the Pauli operator,
\begin{equation}
	\label{eq:pauli}
	P =  \begin{pmatrix}
		H^+ & 0 \\
		0  &   H^-
	\end{pmatrix}\,  , \qquad H^\ppm =  (-i\nabla +A)^2  \pm  B\,.
\end{equation}
Here $A:\R^2\to\R^2$ is a vector field and the function $B:\R^2\to\R$ is defined by
$$
B = {\rm curl}\, A = \partial_1 A_2 - \partial_2 A_1 \,.
$$
For simplicity we restrict ourselves to the case where $V:\R^2\to\R$ is scalar, that is, acts trivially on the $\C^2$ part of $\Lp^2(\R^2,\C^2)$. Both $B$ and $V$ are assumed to be sufficiently regular and to decay in a suitable sense at infinity, as will be made precise later on.

Physically, the operator $P+V$ describes a quantum particle moving in a plane in the presence of a magnetic field of strength $B$ pointing orthogonal to this plane and in the presence of an electric field with potential $V$. The matrix structure of $P$ and the $\pm B$ term in $P$ come from the interaction of the spin of the particle with the magnetic field. This spin-orbit coupling is neglected when considering the magnetic Schr\"odinger operator. This simplifies the model, but has the effect of destroying some of the structure of the Pauli operator. In particular, zero modes are removed and the bottom of the spectrum is stabilized. In our study we will \emph{not} neglect the spin-orbit coupling and we will pay special attention to effects coming from the low energy part of the operator $P$.

When $B$ and $V$ are sufficiently regular and sufficiently fast decaying (we will be more precise later on), the differential expression $P+V$ can be realized as a self-adjoint, lower bounded operator in the Hilbert space $\Lp^2(\R^2,\C^2)$ and the negative spectrum of this operator consists only of eigenvalues with finite multiplicities and with zero as their only possible accumulation point. Labelling these eigenvalues as $E_j$, where multiplicities are taken into account, we are interested in bounding sums
$$
\sum_j |E_j|^\gamma = \tr(P+V)_\m^\gamma
$$
from above for different choices of the parameter $\gamma>0$. These upper bounds shall involve integrals over $\R^2$ of powers of $V$ and quantities defined in terms of the magnetic field $B$. The prototype of such bounds are the Lieb--Thirring inequalities, which in the nonmagnetic case state that for any $\gamma >0$ there is a universal constant $L_\gamma$ such that for all real $V \in \Lp^1_{\rm loc}(\R^2)$ one has
\begin{equation} \label{lt-2D}
\tr (-\Delta +V)_\m^\gamma\, \leq \, L_{\gamma} \int_{\R^2} V(x)_\m^{\gamma+1} \, dx \,.
\end{equation}
Here $a_\ppm:=\max\{\pm a,0\}$, so that $a=a_\pp-a_\m$. The bound \eqref{lt-2D} goes back to the work of Lieb and Thirring \cite{lt2} and has created a huge literature. For further reading on this topic we refer to the monograph \cite{flw-book}, the review \cite{Fr2} and references therein. 

One feature about \eqref{lt-2D} that will be relevant for our discussion is that the inequality gets stronger as $\gamma$ gets smaller. This is formalized by the Aizenman--Lieb argument \cite{al} (see also \cite[Lemma 5.2]{flw-book}), which says that the validity of inequality \eqref{lt-2D} for some $\gamma=\gamma_0$ implies its validity for all $\gamma\geq\gamma_0$.

Turning our attention back to magnetic fields, it is not difficult to see that \eqref{lt-2D} remains valid when $-\Delta$ is replaced by $(-i\nabla+A)^2$; see \cite[Theorem 4.61]{flw-book}, \cite{Fr} and references therein. More precisely, for any $\gamma>0$ there is a constant $\tilde L_\gamma$ such that for any real $V\in \Lp^1_{\rm loc}(\R^2)$ and any $A\in \Lp^2_{\rm loc}(\R^2,\R^2)$ the analogue of \eqref{lt-2D} holds with constant $\tilde L_\gamma$. Note that the right side in the resulting inequality is independent of $A$.

The situation is quite a bit more complicated for the Pauli operator, that is, when the spin-orbit coupling is taken into account. There have been many works addressing this question and we will review them in some detail later in this introduction. For the present discussion the following two features are important. First, there can be no bound of the form \eqref{lt-2D} with a right side that is independent of $A$. Second, previous works are restricted to the range $\gamma\geq 1$. Both phenomena are related to the existence of zero modes of the Pauli operator. The existence of the latter and their structure is described by the Aharonov--Casher theorem~\cite{ac}.

What we shall show in the present paper is that if the normalized magnetic flux
\begin{equation}
	\label{eq:alphaflux}
	\alpha:= \frac{1}{2\pi} \int _{\R^2}B(x)\, dx < \infty 
\end{equation}
satisfies
\begin{equation}
	\label{eq:fluxineq}
	|\alpha|<1 \,,
\end{equation}
then a Lieb--Thirring inequality holds for $P+V$ whenever $\gamma\geq \alpha$ and $\gamma>0$. Moreover, we shall show that this restriction on $\gamma$ is optimal.

According to the Aharonov--Casher theorem \cite{ac} (see also \cite{cfks}), assumption \eqref{eq:fluxineq} implies that the Pauli operator $P$ does not have a zero eigenvalue. Heuristically, this eliminates the reason for the restriction $\gamma\geq 1$ in earlier works. The Pauli operator $P$ does, however, have a zero energy resonance, that this, there is a function $\psi_0$ with $P\psi_0=0$ which decays at infinity, but not fast enough to be square integrable. The decay of this function will be what dictates the optimal condition $\gamma\geq \alpha$ on the exponent in the Lieb--Thirring inequality. Our Lieb--Thirring inequality will have two terms on the right side, the first one being the standard term from \eqref{lt-2D} and the second one involving explicitly the resonance function $\psi_0$.

%%%%%%%%%%%%%%%%%%%%%%%%

\subsection{Definitions and main result}

We now turn to a precise formulation of our result, beginning with a careful definition of the Pauli operator $P$. The standard definition of $P$ assumes that $A\in \Lp^2_{\rm loc}(\R^2,\R^2)$ and proceeds from the quadratic forms
$$
\int_{\R^2} |(\Pi_1+i\Pi_2)\psi^\pp|^2\,dx + \int_{\R^2} |(\Pi_1-i\Pi_2)\psi^\m|^2\,dx
$$
where $\Pi_j :=-i\partial_j + A_j$ and where the form is defined for all $(\psi^\pp,\psi^\m)\in\Lp^2(\R^2,\C^2)$ for which the distributions $(\Pi_1+i\Pi_2)\psi^\pp$ and $(\Pi_1-i\Pi_2)\psi^\m$ belong to $\Lp^2(\R^2)$. We will \emph{not} adapt this definition, although the one we choose is equivalent to this standard definition in situations with enough regularity. The reason is that our assumptions are more naturally formulated in terms of the magnetic field $B$ (and a quantity defined in terms of it) rather than in terms of the vector potential $A$, on which the standard definition is based.

The approach that we follow was promoted by Erd\H{o}s and Vougalter and investigated in detail in their paper \cite{ErVo}. To motivate it, we assume that there is a real function $h\in\Lp^2_{\rm loc}(\R^2)$ such that $A_1=-\partial_2 h$ and $A_2=\partial_1 h$. Then a computation shows that
\begin{align*}
	\int_{\R^2} |(\Pi_1+i\Pi_2)\psi^\pp|^2\,dx & = \int_{\R^2} e^{2h} |(\partial_{1} + i \partial_{2}) e^{-h} \psi^\pp |^2\, dx \,, \\
	\int_{\R^2} |(\Pi_1-i\Pi_2)\psi^\m|^2\,dx & = \int_{\R^2} e^{-2h} |(\partial_{1} - i \partial_{2}) e^{h} \psi^\m |^2\, dx \,.
\end{align*}
The basic idea is to use the right sides to \emph{define} the Pauli operator. Note that if such a function $h$ exists, then $\Delta h = \partial_1 A_2 -\partial_ 2 A_1 = B$.

We now proceed to the actual definition of $P$, following \cite{ErVo}. We assume that $\mu$ is a signed real regular Borel measure on $\R^2$ with $\mu(\{x\})=0$ for all $x\in\R^2$. Then, by \cite[Theorem 2.7]{ErVo} for any $p<2$ there is an $h\in W^{1,p}_{\rm loc}(\R^2)$ such that 
$$
\Delta h = \mu 
\qquad\text{in}\ \R^2 \,.
$$
Fixing any such $h$, the quadratic form
\begin{equation}
	\label{eq:formsev}
	\int_{\R^2} e^{2h} \, |(\partial_{1} + i \partial_{2}) e^{-h} \psi^\pp |^2\, dx  + \int_{\R^2} e^{-2h} \, |(\partial_{1} - i \partial_{2}) e^h \psi^\m |^2\, dx \,,
\end{equation}
defined for all $(\psi^\pp,\psi^\m)\in \Lp^2(\R^2,\C^2)$ for which the integrals are finite, is nonnegative and closed in $\Lp^2(\R^2,\C^2)$ \cite[Theorem 2.5]{ErVo} and therefore generates a selfadjoint, nonnegative operator $P$ in $\Lp^2(\R^2,\C^2)$. This operator depends on the choice of the function $h$, but one can show that for two different choices of functions $h$ the resulting operators are unitarily equivalent by a gauge transformation \cite[Theorem 2.5]{ErVo}. Clearly, for two functions $h$ differing by an additive constant the corresponding operators coincide. Moreover, if there is an $A\in \Lp^2_{\rm loc}(\R^2,\R^2)$ with $\partial_1 A_2 - \partial_2 A_1 = \mu$ (in the sense of distributions), then the operator $P$ is unitarily equivalent to the Pauli operator defined via the standard approach outlined above \cite[Proposition 2.10]{ErVo}. 

We can now formulate our assumptions on the magnetic field. It is formulated in terms of the auxiliary function~$h$ that appears in the definition of the Pauli operator.

\begin{assumption}\label{ass}
	There is an $\alpha\in(-1,1)$ and an $R>0$ such that the two numbers
	\begin{equation} \label{h-bounds}
		m^\ppm :=  \esssup_{x\in \R^2} \frac{e^{\pm h(x)}}{(1+|x|/R)^{\pm \alpha}}
	\end{equation}
	are both finite.
\end{assumption}

The number $R$ will play a minor role in what follows and is only introduced for dimensional consistency. In contrast the number $\alpha$ does play an important role. We emphasize that, if $\mu$ is absolutely continuous with $B=\frac{d\mu}{dx}\in \Lp^1(\R^2)$, then the validity of Assumption~\ref{ass} implies that the number $\alpha$ is necessarily given by the expression \eqref{eq:alphaflux}. We provide a proof of this claim in Lemma~\ref{meaningalpha}.

A simple case where Assumption \ref{ass} is satisfied is $\mu=(\alpha/R)\mathcal H^1_{\partial B(0,R)}$, with $\mathcal H^1_{\partial B(0,R)}$ denoting surface measure on the circle $\partial B(0,R)$ of radius $R$ centered at the origin. In this case we can choose $h=\alpha\ln_\pp(|x|/R)$ and we see that \eqref{h-bounds} is satisfied with the given $\alpha$ and $R$.

We emphasize that, while we can treat $\mu$ that are not absolutely continuous, our main interest is in the absolutely continuous case with $B=\frac{d\mu}{dx}\in \Lp^1(\R^2)$. For instance it is easy to see that if $B$ satisfies
\begin{equation}
	\label{eq:ptwdecayb}
	|B(x)| \leq C R^{-2} (1+|x|/R)^{-\rho}
	\qquad\text{with some}\ C>0 \ \text{and}\ \rho>2 \,,
\end{equation}
then \eqref{h-bounds} holds with $\alpha$ given by \eqref{eq:alphaflux} and the given $R$. The numbers $m^\ppm$ are bounded in terms of $C$ and $\rho$. In Lemma \ref{hlp} we show that Assumption \ref{ass} is satisfied under rather weak integrability assumptions on~$B$.

We are now ready to formulate our main result.

\begin{theorem}\label{main}
	Let Assumption \ref{ass} be satisfied. Then for any $\gamma\geq |\alpha|$ with $\gamma>0$ there are constants $L_1(\gamma,\mu)$ and $L_2(\gamma,\mu)$ such that for every real $V\in \textup{\Lp}^1_{\rm loc}(\R^2)$ one has
	$$
	\tr(P+V)_\m^\gamma \leq L_1(\gamma,\mu) \int_{\R^2} V(x)_\m^{\gamma+1}\,dx + L_2(\gamma,\mu) \int_{\R^2} e^{-2(\sgn\alpha) (h(x)-h_0)}\, V(x)_\m^{\gamma+1-|\alpha|}\,dx \,.
	$$
	Here
	$$
	h_0 := \begin{cases}
		\lim_{\epsilon\to 0} \essinf_{B(0,\epsilon)} h & \text{if}\ \alpha>0 \,, \\
		\lim_{\epsilon\to 0} \esssup_{B(0,\epsilon)} h & \text{if}\ \alpha<0 \,.
	\end{cases}
	$$
	The constants $L_1(\gamma,\mu)$ and $L_2(\gamma,\mu)$ can be chosen such that
	\begin{align*}
		L_1(\gamma,\mu) & \leq C(|\alpha|,\gamma) \, (m^\pp m^\m)^{2(\gamma+1)} \,, \\
		L_2(\gamma,\mu) & \leq C(|\alpha|,\gamma) \, R^{-2|\alpha|} \, (m^\pp m^\m)^{2(\gamma-|\alpha|+2)} \,,
	\end{align*}
	where $C(|\alpha|,\gamma)$ depends only on $|\alpha|$ and $\gamma$.
\end{theorem}

\begin{remarks}
	Some comments on the above theorem are in order.
	\begin{enumerate}
		\item[(a)] When $\alpha\neq 0$ there are two different terms on the right side. These two terms capture the correct order in the strong and weak coupling limit where $V$ is replaced by $\lambda V$ and either $\lambda\to\infty$ or $\lambda\to 0$. Indeed, the first term on the right side grows like $\lambda^{\gamma+1}$ as $\lambda\to\infty$, which is optimal in view of the Weyl asymptotics
		$$
		\lim_{\lambda\to\infty} \lambda^{-1-\gamma} \tr(P+\lambda V)_\m^\gamma = \frac{1}{2\pi\,(\gamma+1)} \, \int_{\R^2} V(x)_\m^{\gamma+1}\,dx \,.
		$$
		In the weak coupling limit with $\gamma=|\alpha|>0$ the second term on the right side vanishes linearly as $\lambda\to 0$, which is optimal since according to \cite{fmv,kov} one has
		\begin{equation} \label{weak-coupling}
			\lim_{\lambda \to 0_\pp}\,  \lambda^{- \frac\gamma{|\alpha|}}\,  \tr (P + \lambda V)_\m^{\gamma} = \left( - \frac{4^{|\alpha|-1}\, \Gamma(|\alpha|)}{ \pi\, \Gamma(1-|\alpha|)}\ \int_{\R^2} V(x) \, e^{-2(\sgn\alpha) h(x)} dx \right)^\frac\gamma\alpha,
		\end{equation}
		provided the integral on the right side is nonpositive and $h$ is chosen in a certain canonical way. We emphasize that this argument also shows that the function $e^{-2(\sgn\alpha)h}$ in our bound captures quantitatively the relevant quantity in the weak coupling limit. 
		
		\item[(b)] The assumption $\gamma\geq |\alpha|$ for $\alpha\neq 0$ is optimal. Indeed, by the weak coupling asymptotics \eqref{weak-coupling}, $\tr(P+\lambda V)_\m^\gamma$ behaves like $\lambda^\frac{\gamma}{|\alpha|}$, while the second term on the right side behaves like $\lambda^{1+\gamma-|\alpha|}$. This shows that for $0<|\alpha|<1$ the assumption $\gamma\geq |\alpha|$ is necessary. Similarly, in \cite{fmv,kov} there are weak coupling asymptotics for $\alpha=0$, which show that in this case the assumption $\gamma>0$ is necessary.
		\item[(c)] Concerning the condition $|\alpha|<1$ in Assumption \ref{ass} we remark that our bound cannot hold for $|\alpha|>1$. This follows again from weak coupling asymptotics in \cite{fmv,kov}, which state that $\tr(P+\lambda V)_\m^\gamma$ behaves like $\lambda^\gamma$ when $|\alpha|\geq 1$. Meanwhile, the second term on the right side of our bound behaves like $\lambda^{1+\gamma-|\alpha|}$, showing that the bound can only hold when $|\alpha|\leq 1$. This leaves open the case $|\alpha|=1$ for which one might expect a bound for $\gamma\geq 1$. As discussed in the next subsection, under somewhat different assumptions on the magnetic field such a bound was indeed shown in \cite{sob}, which is why we did not investigate it further.		
		\item[(d)] The function $e^{-(\sgn\alpha)(h-h_0)}$ coincides, up to a phase factor, with the zero energy resonance function $\psi_0$ mentioned in Subsection \ref{ssec-background}. Also, since the operator $P$ does not change if a constant is added to $h$, the difference $h-h_0$ that appears in our bound is indeed a natural quantity. When $h$ is continuous at the origin, we clearly have $h_0=h(0)$. The particular way of how to define $h_0$ in the discontinuous case is dictated mostly by technical convenience. We emphasize that $h_0$ is finite in view of \eqref{h-bounds}. The fact that the point $0$ is singled out in the definition of $h_0$ reflects that this point is singled out in \eqref{h-bounds}.
		\item[(e)] Our bound depends on the `magnetic field' $\mu$ only via the function $h$ and this dependence is only via the quantities $\alpha$, $R$ and $m^\ppm$ from Assumption \ref{ass}. In particular, note that
		$$
		m^\pp m^\m = \esssup_{x\in \R^2} \frac{e^{h(x)}}{(1+|x|/R)^{\alpha}}
		\esssup_{x\in \R^2} \frac{e^{- h(x)}}{(1+|x|/R)^{- \alpha}} \geq 1\,.
		$$
		In the weak field limit where $B$ is replaced by $\lambda B$ and $\lambda\to 0$, the function $h$ is replaced by $\lambda h$ and $\alpha$ by $\lambda\alpha$, while $R$ remains unchanged. The product $m^\pp m^\m$ is replaced by $(m^\pp m^\m)^\lambda$, which tends to 1. Our proof will show that for fixed $\gamma>0$, the constant $C(|\lambda\alpha|,\gamma)$ remains bounded as $\lambda\to 0$; see Remark \ref{smallfieldrem}. Thus, our bound is stable in the limit $\lambda\to 0$ and reproduces the nonmagnetic Lieb--Thirring inequality \eqref{lt-2D}. This property is not shared, for instance, by the bound from \cite{sob} discussed in the next subsection.
		\item[(f)] We have been somewhat cavalier about our assumptions on $V$. Here is a more precise statement: If $V\in\Lp^1_{\rm loc}(\R^2)$ is real and if the right side in the bound in the theorem is finite, then $V_\m$ is infinitesimally form bounded with respect to $P$ and for the operator $P+V$, defined via quadratic forms, the stated bound holds. This follows by standard argument from our proof. The same statement holds for all Lieb--Thirring-type inequalities in this paper and will not be repeated each time.
		\item[(g)] Assume that $\mathcal V$ is a locally integrable function on $\R^2$ taking values in the Hermitian $2\times 2$-matrices. Then the corresponding inequality holds for the operator $P+\mathcal V$, provided on the right side we replace $V(x)_\m^p$ by $\tr_{\C^2} (\mathcal V(x)_\m^p)$ for $p\in\{1+\gamma,1+\gamma-|\alpha|\}$. This simply follows from the inequality $\mathcal V(x) \geq - \|\mathcal V(x)_\m\|$ (with $\|\cdot\|$ the operator norm on $\C^2$), our bound in the scalar case and the bound $\|\mathcal V(x)_\m\|^p \leq \tr_{\C^2} (\mathcal V(x)_\m^p)$. For this reason we restrict ourselves to the case of a scalar electric potential.
		\item[(h)] Our inequality comes with explicit values for the constants $C(|\alpha|,\gamma)$, but since they are far from optimal we do not state them explicitly.
	\end{enumerate}
\end{remarks}

%%%%%%%%%%%%%%%%%%%%%%%%%%%%%
%%%%%%%%%%%%%%%%%%%%%%%%%%%%%

\subsection{Previous results}

Let us review some previous works on Lieb--Thirring inequalities for Pauli operators and compare them with our new results. Throughout we focus on the two-dimensional case and leave out many important advances in the three dimensional case, starting with Erd\H{o}s's foundational work \cite{Er} and reviewed in \cite{Er1,BlFo}.

Lieb, Solovej and Yngvason \cite{LiSoYn} showed that when $B$ is constant, then
$$
\tr (P +V)_\m \   \leq\  C \, \Big ( \int_{\R^2} V(x)_\m^2\, dx + |B|  \int_{\R^2}   V(x)_\m \,dx \Big ) \,.
$$
This was generalized by Erd\H{o}s and Solovej \cite[Thm.~3.2]{es} (based on the strategy in \cite{Er}), who showed that for any $\gamma\geq 1$ there is a constant $C_\gamma$ such that 
\begin{equation}  \label{es} 
	\tr (P + V)_\m^\gamma\   \leq\  C_\gamma \, \Big ( \int_{\R^2} V(x)_\m^{\gamma+1}\, dx + \|B\|_\infty  \int_{\R^2}   V(x)_\m^\gamma\, dx \Big ) \,.
\end{equation}

More relevant for us is an earlier work of Sobolev \cite{sob}, where it was shown under fairly general conditions on $B$ that for any $\gamma\geq 1$ there is a 
constant $C_{\gamma}$ such that 
\begin{align}  \label{sobolev} 
	\tr (P + V)_\m^\gamma\  & \leq\  C_{\gamma} \, \Big( \int_{\R^2} V(x)_\m^{\gamma+1}\, dx  +  \int_{\R^2} b(x)\,  V(x)_\m^{\gamma}\, dx \Big) \,.
\end{align}
Here $b$ denotes a ``smeared" modification of $B$, see \cite[Sec.~2]{sob} for details. It should be noted that $b$ is not uniquely defined. The relevance of a smeared magnetic field was pointed out in \cite{Er}.

The assumptions on $B$ in \cite{sob} are somewhat implicit. In order to compare the results with ours, we assume that $B$ satisfies the pointwise decay condition \eqref{eq:ptwdecayb}. Then it is easily verified that an effective magnetic field $b$ in the sense of \cite{sob} can be constructed in such a way as to satisfy 
$$
b(x) \ \leq \ C_B \,  (1+|x|)^{-2}\, \qquad \forall\, x\in\R^2,
$$
see \cite[Eqs.~(2.7)-(2.11)]{sob}. Inequality \eqref{sobolev} then implies that for any $\gamma\geq 1$,
\begin{equation} \label{sobolev-bis}
	\tr (P + V)_\m^\gamma\  \leq\   C_{1,\gamma} \int_{\R^2} V(x)_\m^{\gamma+1}\, dx  + C_{2,\gamma}(B) \int_{\R^2} (1+|x|)^{-2} \, V(x)_\m^{\gamma}\, dx \,.
\end{equation}

Let us compare \eqref{es} and \eqref{sobolev-bis} with our bound in Theorem \ref{main}. Importantly, \eqref{es} and \eqref{sobolev-bis} do not have a restriction on the normalized flux $\alpha$ of $B$. Meanwhile, they are restricted to values $\gamma\geq 1$. When $|\alpha|<1$, H\"older's inequality yields that
$$
\int_{\R^2} (1+|x|)^{-2\alpha} V(x)_\m^{\gamma+1-|\alpha|} \,dx \leq \left( \int_{\R^2} (1+|x|)^{-2} V(x)_\m^\gamma\,dx \right)^{|\alpha|}
\left( \int_{\R^2} V(x)_\m^{\gamma+1}\,dx \right)^{1-|\alpha|} \,.
$$
Since $e^{-2h}\leq (m^\m)^2 (1+|x|)^{-2\alpha}$, we see that our bound in Theorem \ref{main} implies \eqref{sobolev-bis} for $|\alpha|<1$.

In the strong coupling regime, where $V$ is replaced by $\lambda V$ with $\lambda\to\infty$, the bounds \eqref{es}, \eqref{sobolev-bis} and our bound all reproduce the optimal $\lambda^{\gamma+1}$ growth.

In the weak coupling regime, where $V$ is replaced by $\lambda V$ with $\lambda\to 0_+$, $\tr (P + \lambda V)_\m^\gamma$ vanishes like $\lambda^\gamma$ when $|\alpha|\geq 1$, so both \eqref{es} and \eqref{sobolev-bis} are order-sharp in this case. However, $\tr (P + \lambda V)_\m^\gamma$ vanishes like $\lambda^\frac{\gamma}\alpha$ when $0<|\alpha|<1$, and in this regime \eqref{es} and \eqref{sobolev-bis} are no longer order-sharp, while the bound from Theorem \ref{main} is.

Concerning the regime of a weak magnetic field, where $B$ is replaced by $\lambda B$ with $\lambda\to 0$, we see that \eqref{es} turns into the ordinary Lieb--Thirring inequality (for $\gamma\geq 1$), as does the bound in Theorem \ref{main} (for $\gamma>0$). Meanwhile, as pointed out in \cite[p.~614]{sob}, inequalities \eqref{sobolev} and \eqref{sobolev-bis} are not applicable in the regime of the weak magnetic field.

These observations suggest that Theorem \ref{main} improves over both \eqref{es} and \eqref{sobolev-bis} in the small flux regime $|\alpha|<1$. It displays not only the sharp behavior in the strong coupling limit, but also in the weak coupling limit and allows for a smooth passage to inequality \eqref{lt-2D} in the limit of a vanishing magnetic field. This is made possible by replacing the integrand $(1+|x|)^{-2} V(x)_\m^\gamma$ in \eqref{sobolev-bis} by $(1+|x|)^{-2|\alpha|} V(x)_\m^{\gamma+1-|\alpha|}$.

\medskip

It is also interesting to view our results from the point of view of Lieb--Thirring inequalities in the presence of zero energy resonances (aka virtual levels). Lieb--Thirring inequalities when a critical Hardy weight is subtracted from the Laplacian were shown in \cite{EkFr0,ef,FrLiSe,Fr1}. In these cases, like in the present one, there is an algebraically decaying zero energy resonance function. The resulting inequality, however, only has a single term in contrast to our bound for $P$ when $\alpha\neq0$, which has two terms. This is connected with the fact that eigenvalues for the Hardy operator are exponentially small in the limit of a vanishing coupling constant. Hardy--Lieb--Thirring inequalities for fractional Pauli operators in three dimensions were studied in \cite{BlFo}.

Lieb--Thirring inequalities in the presence of a resonance function that is bounded from above and away from zero were studied in \cite{FrSiWe}. The resulting inequality has only a single term. This is relevant in the present case in the simplest case $\alpha=0$.

In \cite{EkFrKo,FrKo} we investigated Lieb--Thirring inequalities in the context of Schr\"odinger operators on continuous graphs that are sparse in some sense. The Lieb--Thirring inequalities in this case have two terms, reflecting the different behavior in the strong and the weak coupling limit; for the latter see \cite{Ko}. Compared with \cite{FrKo} the results in the present paper are substantially more precise, as we are able to prove the Lieb--Thirring inequality in the critical case $\gamma=|\alpha|>0$, while the corresponding question is left open in \cite{FrKo}. Nevertheless some techniques from \cite{FrKo} will play a role in our analysis of subcritical cases; see Section \ref{sec:subcritical}.

Our result also shows similarities to the logarithmic Lieb--Thirring inequality for the two-dimensional Schr\"o\-din\-ger operator \cite{KoVuWe}, where again two terms appear on the right side. Our proof in the critical case uses some ideas from \cite{KoVuWe} (and \cite{EkFr0}); see Section \ref{sec:critical}. An important conceptual difference, however, is that in the relevant bound on the lowest eigenvalue in Proposition \ref{prop-E-1} still two terms appear while there is only one term in the corresponding bound in \cite[Lemma 1]{KoVuWe}. This leads to substantial technical difficulties that need to be overcome; see Appendix \ref{sec-one-term} for a proof of the fact that both terms are necessary.

\medskip

Finally, we mention the works \cite{weidl,fmv,kov,BrKo,FaKo} that quantify different aspects of the instability of the bottom of the spectrum of the Pauli operator. Some of the techniques developed there will be relevant for us here.

%%%%%%%%%%%%%%%%%%%%%%%%%%%%%%%%%%%%
%%%%%%%%%%%%%%%%%%%%%%%%%%%%%%%%%%%%

\subsection{Strategy of the proof}\label{sec:strategy}

Since the Pauli operator is block diagonal, Theorem \ref{main} is an immediate consequence of two theorems concerning the individual blocks. These two operators are defined by the first and the second quadratic form on the right side of \eqref{eq:formsev} and are denoted by $H^+$ and $H^-$, respectively. They are operators acting in the space $\Lp^2(\R^2)$ of complex-valued functions. When $\mu$ is absolutely continuous with sufficiently regular density $B=\frac{d\mu}{dx}$, these operators coincide with those given by \eqref{eq:pauli}.

The block-diagonality of $P$ and the spin-independence of $V$ imply
\begin{equation}
	\label{eq:blockdiagonal}
	\tr_{\, \Lp^2(\R^2,\C^2)}(P+V)_\m^\gamma = \tr_{\, \Lp^2(\R^2)}(H^\pp +V)_\m^\gamma \, + \, \tr_{\,\Lp^2(\R^2)}(H^\m +V)_\m^\gamma \,.
\end{equation}
In what follows we solely discuss the operators $H^\pp$ and $H^\m$, rather than $P$.

It will turn out that for $\alpha\neq 0$ only one of the two operators $H^+$ and $H^-$ is `critical' while the other one is `subcritical'. (One could give a mathematical definition of what we mean by `critical' and `subcritical', but since we do not need anything from the corresponding theory, we will use these terms only in a colloquial sense and refer to \cite{weidl,Pi} for some background.) In order to discuss the distinction between $H^+$ and $H^-$, we shall assume that
$$
\alpha\geq 0 \,.
$$ 
This is no loss of generality, since replacing $\mu$ by $-\mu$ can be compensated by replacing $h$ by $-h$ and then replacing $\alpha$ by $-\alpha$ in \eqref{h-bounds}. Of course this is also consistent with the expression \eqref{eq:alphaflux} for $\alpha$ in the regular case. The product $m^\pp m^\m$ that appears in our bounds is invariant under this replacement. 

With this convention in place, the operator $H^-$ is `critical', while $H^+$ is `subcritical' for $\alpha>0$. The following result says that for the subcritical operator a Lieb--Thirring inequality holds for arbitrarily small $\gamma>0$.

\begin{theorem}\label{thm-main-1} 
	Let Assumption \ref{ass} be satisfied with $\alpha\geq 0$. Then for any $\gamma>0$ there is a constant $L(\gamma,\mu)$ such that for every real $V\in \textup \Lp^1_{\rm loc}(\R^2)$ one has
	$$
	\tr(H^+ + V)_\m^\gamma \leq L(\gamma,\mu) \int_{\R^2} V(x)_\m^{\gamma+1}\,dx \,.
	$$
	The constant $L(\gamma,\mu)$ can be chosen such that
	\begin{align*}
		L(\gamma,\mu) & \leq C(\alpha,\gamma) \, (m^\pp m^\m)^{2(\gamma+1)} \,,
	\end{align*}
	where $C(\alpha,\gamma)$ depends only on $\alpha$ and $\gamma$. The same assertion holds for the operator $H^-$ if $\alpha=0$.
\end{theorem}

\begin{theorem} \label{thm-main-3}
	Let Assumption \ref{ass} be satisfied with $\alpha> 0$. Then for any $\gamma\geq\alpha$ there are constants $L_1(\gamma,\mu)$ and $L_2(\gamma,\mu)$ such that for every real $V\in \textup\Lp^1_{\rm loc}(\R^2)$ one has
	$$
	\tr(H^- + V)_\m^\gamma \leq L_1(\gamma,\mu) \int_{\R^2} V(x)_\m^{\gamma+1}\,dx + L_2(\gamma,\mu) \int_{\R^2} e^{-2 (h(x)-h_0)}\, V(x)_\m^{\gamma+1-\alpha}\,dx \,.
	$$
	The constants $L_1(\gamma,\mu)$ and $L_2(\gamma,\mu)$ can be chosen such that
	\begin{align*}
		L_1(\gamma,\mu) & \leq C(\alpha,\gamma) \, (m^\pp m^\m)^{2(\gamma+1)} \,, \\
		L_2(\gamma,\mu) & \leq C(\alpha,\gamma) \, R^{-2\alpha} \, (m^\pp m^\m)^{2(\gamma-\alpha+2)} \,,
	\end{align*}
	where $C(\alpha,\gamma)$ depends only on $\alpha$ and $\gamma$.
\end{theorem}

As we already mentioned, in view of \eqref{eq:blockdiagonal}, Theorem \ref{main} is an immediate consequence of Theorems \ref{thm-main-1} and \ref{thm-main-3}. Most of the remarks following Theorem \ref{main} have analogues for Theorems \ref{thm-main-1} and \ref{thm-main-3}, showing in particular their optimality. We omit the details.

We will prove Theorems \ref{thm-main-1} and \ref{thm-main-3} only in the case
$$
R=1 \,.
$$
This is no loss of generality, according to the following simple scaling argument. If $\mu$ satisfies Assumption \ref{ass} for some $R$ and $\alpha$, then we can define a measure $\tilde\mu$ on $\R^2$ that satisfies Assumption \ref{ass} with the same $\alpha$, but with $R=1$. For absolutely continuous $\mu$ the corresponding densities $B$ and $\tilde B$ are related by $\tilde B(y) := R^{2} B(Ry)$, and this relation is extended in the natural sense to measures. When passing from $\mu$ to $\tilde\mu$, the function $h$ is replaced by the function $\tilde h(y) := h(Ry)$, which proves our claim about Assumption \ref{ass}. Denoting by $\tilde H^\ppm$ the operators corresponding to $\tilde\mu$, we see that the operators $H^\ppm+V$ are unitarily equivalent to the operators $R^{-2}(\tilde H^\ppm +\tilde V)$ with $\tilde V(y):=R^{2} V(Ry)$. As a consequence, Theorems \ref{thm-main-1} and \ref{thm-main-3} for $\tilde H^\ppm +\tilde V$ (with $R=1$) imply the corresponding theorems for the original operator $H^\ppm +V$ (with arbitrary $R>0$).

%%%%%%%%%%%%%%%%%%%%%%%%%%%%%%%%%%%%
%%%%%%%%%%%%%%%%%%%%%%%%%%%%%%%%%%%%

\section{Passage to weighted spaces} 

In this section we will show that Lieb--Thirring inequalities for $H^\ppm$ follow from corresponding Lieb--Thirring inequalities for certain operators $\mathcal H^\ppm$ that act in a weighted $\Lp^2$ space and are defined through a weighted Dirichlet integral. For this argument it is crucial that $\alpha< 1$.

\subsection{Lower bound on $Q^\ppm$} 

Let
$$
Q^\ppm[\psi] := \int_{\R^2} e^{\pm 2h} \, |(\partial_{1} \pm i \partial_{2}) e^{\mp h} \psi |^2\, dx
$$
denote the quadratic form of the operator $H^\ppm$.

The simple pointwise bound $|(\partial_{x_1} \pm i \partial_{x_2}) \phi |^2 \leq 2\,    |\nabla  \phi |^2$, together with \eqref{h-bounds} (recall our convention $R=1$), shows that
\begin{equation}
	\label{eq:weightedupper}
	Q^\ppm[e^{\pm h} \phi] \leq 2 \, (m^\ppm)^2 \int_{\R^2} (1+|x|)^{\pm 2\alpha} |\nabla \phi |^2\,dx \,.
\end{equation}
This holds irrespectively of the value of $\alpha$, as long as \eqref{h-bounds} is valid. The following proposition shows that, under the assumption $|\alpha|<1$, the reverse bound holds, up to changing the value of the constant. This will be one of the main technical tools in the proof of our results.

\begin{proposition} \label{prop-lowerb}
	Let Assumption \ref{ass} be satisfied with $\alpha\geq 0$ and $R=1$. Then for all $\phi\in C_c^1(\R^2)$,
	\begin{equation}\label{eq:weightedlower}
		Q^\ppm[e^{\pm h} \phi] \geq q_\alpha (m^\mpp)^{-2} \int_{\R^2} (1+|x|)^{\pm 2\alpha} |\nabla \phi|^2\,dx \,,
	\end{equation}
where
\begin{equation}  \label{q-alpha}
q_\alpha = \frac{2^{-2\alpha-1}(1-\alpha)^2}{2\alpha + 2^{-2\alpha-1}(1-\alpha)^2} \,.
\end{equation}
\end{proposition}

To prove Proposition \ref{prop-lowerb} we will need some classical results on doubly weighted one-dimensional Hardy inequalities. For the proof we refer to \cite{muck, tom}, see also \cite{flw} and references therein.

\begin{lemma}\label{thm-class-1}
Let $U,W$ be nonnegative, measurable functions on $(0,\infty)$ and let $f$ be a locally absolutely continuous function on $(0,\infty)$. Then the inequality 
\begin{equation} 
	\int_0^\infty W(t)\, |f(t)|^2\, dt \, \leq\, C(U,W) \, \int_0^\infty U(t)\, |f'(t)|^2\, dt
\end{equation} 
holds
\begin{enumerate}
	\item[(a)] if $\liminf_{t\to\infty} |f(t)|=0$ with
	\begin{align} 
		C(U,W) & = 4 \, \sup_{s>0} \Big (\int_s^\infty U(t)^{-1}\, dt \Big)  \Big (\int_0^s W(t)\, dt \Big) . \label{C-upperb-1} 
	\end{align}
	\item[(b)] if $\liminf_{t\to 0} |f(t)|=0$ with
	\begin{align} 
		C(U,W) & = 4 \, \sup_{s>0} \Big (\int_0^s U(t)^{-1}\, dt \Big)  \Big (\int_s^\infty W(t)\, dt \Big) . \label{C-upperb-3}
	\end{align} 
\end{enumerate}
\end{lemma} 

We now turn to the proof of the main result of this section. The argument has some similarities with one used in \cite{FaKo}, but our focus is different.

\begin{proof}[Proof of Proposition \ref{prop-lowerb}]
	By the bounds \eqref{h-bounds}, we have
	$$
	Q^\ppm[e^{\pm h} \phi] \geq (m^\mpp)^{-2} \int_{\R^2} (1+|x|)^{\pm 2\alpha} |(\partial_{1} \pm i \partial_{2}) \phi |^2\, dx \,.
	$$
	For $\alpha=0$, the assertion follows immediately from the fact that
	$$
	\int_{\R^2} |(\partial_{1} \pm i \partial_{2}) \phi |^2\, dx
	= \int_{\R^2} |\nabla \phi |^2\, dx \,.
	$$
	
	For $\alpha>0$ we introduce polar coordinates $x=(r\cos\theta,r\sin\theta)$ and expand $\phi(r\,\cdot)$ into a Fourier series,
	\begin{equation*}
		\phi(x) = \sum_{m\in\Z} e^{im\theta}\, \phi_m(r) \,,
		\qquad
		\phi_m(r) = \frac{1}{2\pi} \int_0^{2\pi}\!\! e^{-im\theta}\, v(r,\theta)\, d\theta \,.
	\end{equation*}
	A computation (see also \cite[Section~10]{weidl}) shows that
\begin{equation}\label{vm-identity1}
\int_0^{2\pi} |(\partial_{1} \pm i \partial_{2}) \phi |^2\, d\theta=  \sum_{m\in\Z}  \Big | \phi'_m(r) \mp \frac{m\, \phi_m(r)}{r} \Big|^2
\end{equation} 
and
\begin{equation}\label{vm-identity2}
\int_0^{2\pi} |\nabla \phi |^2\, d\theta = \sum_{m\in\Z} \left( |\phi_m'(r)|^2 + \frac{m^2}{r^2} |\phi_m(r)|^2 \right).
\end{equation}
Thus, the assertion will follow if we can prove that for each $m\in\Z$,
\begin{equation}
	\label{eq:weightedproofgoal}
	\int_0^\infty (1+r)^{\pm 2\alpha} \,\Big | \phi'_m(r) \mp \frac{m\, \phi_m(r)}{r} \Big|^2 r\,dr \geq q_\alpha \int_0^\infty (1+r)^{\pm 2\alpha} \left( |\partial_r \phi_m(r)|^2 + \frac{m^2}{r^2} |\phi_m(r)|^2 \right) r\,dr \,.
\end{equation}

Integrating by parts we find
\begin{equation}
	\label{eq:weightedproofibp}
	\int_0^\infty (1+r)^{\pm 2\alpha} \,\Big | \phi'_m(r) \mp \frac{m\, \phi_m(r)}{r} \Big|^2 r\,dr
	= \int_0^\infty (1+r)^{\pm 2\alpha}\,   \Big( |\phi'_m|^2 + \frac{m^2\, |\phi_m|^2}{r^2} +2\alpha m \frac{ |\phi_m|^2}{(1+r)r} \Big)\, r\, dr \,.
\end{equation}
When $m\geq 0$ the last term on the right side is nonnegative and we arrive at \eqref{eq:weightedproofgoal}, even with constant $1$ instead of $q_\alpha$. (Note that $q_\alpha\leq 1$.)

From now on we assume that $m\leq -1$. The basic idea is to prove a Hardy inequality that allows us to absorb the last term on the right side of \eqref{eq:weightedproofibp} into the left side. We will apply Lemma \ref{thm-class-1} with 
$$
f(r) = r^{\mp m} \phi_m(r) \,,
\qquad
U(r) = r^{\pm 2m+1} (1+r)^{\pm 2\alpha} \,,
\qquad
W(r) = r^{\pm 2m-1} (1+r)^{\pm 2\alpha} \,. 
$$
Note that the left side of \eqref{eq:weightedproofibp} is equal to $\int_0^\infty U(r) |f'(r)|^2\,dr$. In order to bound the constant in Lemma \ref{thm-class-1}, we distinguish two cases according to the sign.

\medskip

\emph{Case of the upper sign.} In this case we have $\liminf_{t\to 0} |f(t)|=0$, so we aim at applying part (b) of Lemma \ref{thm-class-1}. We have
\begin{equation} \label{V-1}
\int_0^s U(t)^{-1}\, dt  = \int_0^s t^{-2m-1} (1+t)^{-2\alpha}\, dt \, \leq \, \frac{s^{-2m}}{2|m|}\,  \1_{(0,1]}(s) +  \frac{s^{-2m-2\alpha}}{2(|m|-\alpha)} \, \1_{(1,\infty)}(s)
\end{equation}
and
\begin{equation*} 
\int_s^\infty W(t)\, dt  = \int_s^\infty  t^{-1+2m} (1+t)^{2\alpha}\, dt \, \leq \, 2^{2\alpha}\ \Big[  \frac{s^{2m}}{2|m|} +\frac{1}{2(|m|-\alpha)} \Big] \,\1_{(0,1]}(s) +2^{2\alpha} \frac{s^{2m+2\alpha}}{2(|m|-\alpha)} \,\1_{(1,\infty)}(s) . 
\end{equation*}
Hence 
\begin{align*} 
 \sup_{0<s\leq 1} \Big (\int_0^s U(t)^{-1}\, dt \Big)  \Big (\int_s^\infty W(t)\, dt \Big)  &\leq\,   \frac{2^{2\alpha}}{4m^2} +  \frac{2^{2\alpha}}{4|m| (|m|-\alpha)} \, \leq\,   \frac{2^{2\alpha+1}}{4m^2(1-\alpha)} \, ,
\end{align*}
where we have used the elementary bound
\begin{equation*} %\label{m^2-estim}
	(k-\alpha)^2 \,\geq \, k^2\, (1-\alpha)^2 \qquad 0 <\alpha <1, \qquad k\in\Z.
\end{equation*}
Similarly, 
\begin{align*} 
 \sup_{1<s<\infty} \Big (\int_0^s U(t)^{-1}\, dt \Big)  \Big (\int_s^\infty W(t)\, dt \Big)  &\leq\,   \frac{2^{2\alpha}}{4(|m|-\alpha)^2} \, \leq\,   \frac{2^{2\alpha}}{4m^2(1-\alpha)^2} \, .
\end{align*}
 Altogether we deduce from Lemma \ref{thm-class-1} that 
\begin{equation} \label{hardy-m-bigger-alpha}
\int_0^\infty (1+r)^{2\alpha}\,  \,\Big | \phi'_m - \frac{m\, \phi_m}{r} \Big|^2 \, r\,dr\,  \geq\, 2^{-2\alpha-1}\,  (1-\alpha)^2 \, m^2  \int_0^\infty \frac{(1+r)^{2\alpha}}{r^2}\, |\phi_m|^2\, r dr \,.
\end{equation}

\medskip

\emph{Case of the lower sign.}
In this case we have $\liminf_{t\to\infty} |f(t)|=0$, so we aim at applying part (a) of Lemma \ref{thm-class-1}. Note that $\int_s^\infty U(t)^{-1}\,dt$ in the present case coincides with $\int_s^\infty W(t)\,dt$ in the case of the upper sign and similarly $\int_0^s W(t)\,dt$ in the present case coincides with $\int_0^s U(t)^{-1}\,dt$ in the case of the upper sign. Therefore we obtain from the previous bounds
\begin{align*} 
	\sup_{0<s\leq 1} \Big (\int_s^\infty U(t)^{-1}\, dt \Big)  \Big (\int_0^s W(t)\, dt \Big)  & \leq\,   \frac{2^{2\alpha+1}}{4m^2(1-\alpha)} \, ,
\end{align*}
and
\begin{align*} 
	\sup_{1<s<\infty} \Big (\int_s^\infty U(t)^{-1}\, dt \Big)  \Big (\int_0^s W(t)\, dt \Big)  &\leq\, \frac{2^{2\alpha}}{4m^2(1-\alpha)} \, .
\end{align*}
Altogether we deduce from Lemma \ref{thm-class-1} that
\begin{equation} \label{hardy-m-bigger-alphab}
	\int_0^\infty (1+r)^{-2\alpha}\,  \,\Big | \phi'_m + \frac{m\, \phi_m}{r} \Big|^2 \, r\,dr\,  \geq\, 2^{-2\alpha-1}\,  (1-\alpha)^2 \, m^2  \int_0^\infty \frac{(1+r)^{-2\alpha}}{r^2}\, |\phi_m|^2\, r dr \,.
\end{equation}

\medskip

\emph{Conclusion of the proof.}
We combine the integration by parts identity \eqref{eq:weightedproofibp} with the Hardy inequalities \eqref{hardy-m-bigger-alpha} and \eqref{hardy-m-bigger-alphab} and obtain, for any $\vartheta\in[0,1]$,
\begin{align*}
	& \int_0^\infty (1+r)^{\pm 2\alpha} \,\Big | \phi'_m(r) \mp \frac{m\, \phi_m(r)}{r} \Big|^2 r\,dr
	 \geq (1-\vartheta) \int_0^\infty (1+r)^{\pm 2\alpha}\,   \Big( |\phi'_m|^2 + \frac{m^2\, |\phi_m|^2}{r^2} \Big)\, r\, dr \\
	& \qquad \qquad \qquad \qquad \quad + \left( (1-\vartheta)2\alpha m + \vartheta 2^{-2\alpha-1}\,  (1-\alpha)^2 \, m^2 \right) 
	\int_0^\infty (1+r)^{\pm 2\alpha}\, \frac{ |\phi_m|^2}{(1+r)r} \, r\, dr \,.
\end{align*}
Here in the Hardy inequalities, we estimated $r^{-2}\geq (r(1+r))^{-1}$. We now choose
$$
\vartheta = \frac{2\alpha|m|}{2\alpha|m| + 2^{-2\alpha-1}(1-\alpha)^2 m^2} \,,
$$
so that the last term vanishes. The constant in front of the first term is equal to
$$
1-\vartheta = \frac{2^{-2\alpha-1}(1-\alpha)^2 m^2}{2\alpha|m| + 2^{-2\alpha-1}(1-\alpha)^2 m^2} \,.
$$
Since this is monotone increasing in $|m|$, a lower bound is obtained by setting $m=-1$, which gives the constant $q_\alpha$. This proves \eqref{eq:weightedproofgoal}.
\end{proof}

\begin{remark}
	The assumption $\alpha<1$ in Proposition \ref{prop-lowerb} is optimal. Indeed, the inequality 
\begin{equation} \label{lowerb-fail}
 \int_{\R^2}  (1+|x|)^{-2\alpha}\, |(\partial_{1} - i \partial_{2}) \phi |^2\, dx  \, \geq\, c \int_{\R^2} (1+|x|)^{-2\alpha}\, |\nabla \phi|^2\, dx \qquad \forall\, \phi\in C^1_c(\R^2)
\end{equation}
fails to hold, for any $c>0$, as soon as $\alpha\geq 1$. Indeed, by density it would then also hold for the functions $\phi^{(R)}$, $R>0$, given by 
$$
\phi^{(R)}(r\cos\theta, r\sin\theta) := 
\left\{
\begin{array}{l@{\quad}l}
 r e^{-i\theta} & \text{if}\ r \leq R \,,   \\[5pt]
\frac{R^2}{r} \, e^{-i\theta} & \text{if}\ R  <  r\,  .
\end{array}
\right. 
$$  
A short calculation using \eqref{vm-identity1} and \eqref{vm-identity2}, however, shows that 
$$
\lim_{R\to \infty}  \frac{ \int_{\R^2}  (1+|x|)^{-2\alpha}\, |(\partial_{1} - i \partial_{2}) \phi^{(R)} |^2\, dx}{\int_{\R^2} (1+|x|)^{-2\alpha}\, |\nabla \phi^{(R)}|^2\, dx} = 0 \,,
$$
which obviously contradicts \eqref{lowerb-fail}.

Meanwhile, inequality \eqref{eq:weightedlower} \emph{with the upper sign} can be extended to all $\alpha$ such that $1 < \alpha \not\in\Z$. Since we will not use this bound, we omit its proof.
\end{remark} 

%%%%%%%%%%%%%%%%%%%%%%%%%%%%%%%%%%%
%%%%%%%%%%%%%%%%%%%%%%%%%%%%%%%%%%%

\subsection{Equivalence of quadratic forms}  

So far we have worked with the operators $H^\ppm$ in the space $\Lp^2(\R^2)$. Now we pass to certain operators $\mathcal H^\ppm$ in the weighted spaces $\Lp^2(\R^2,(1+|x|)^{\pm 2\alpha}\,dx)$ and show that Lieb--Thirring inequalities for the new operators imply Lieb--Thirring inequalities for the original operators.

We consider the quadratic form
$$
\int_{\R^2} (1+|x|)^{\pm 2\alpha}\, |\nabla \phi|^2\, dx
$$
in the Hilbert space $\Lp^2(\R^2,(1+|x|)^{\pm 2\alpha}\,dx)$. The form domain consists of functions $\phi\in H^1_{\rm loc}(\R^2)\cap \Lp^2(\R^2,(1+|x|)^{\pm 2\alpha}\,dx)$ for which the form is finite. It is easy to see that this form is closed in $\Lp^2(\R^2,(1+|x|)^{\pm 2\alpha}\,dx)$. We denote the resulting selfadjoint, nonnegative operator in $\Lp^2(\R^2,(1+|x|)^{\pm 2\alpha}\,dx)$ by $\mathcal H^\ppm$.

From Proposition \ref{prop-lowerb} we deduce the following upper bound on the Riesz means that we are interested in.

\begin{corollary}\label{comparison}
	Let Assumption \ref{ass} be satisfied with $\alpha\geq 0$ and $R=1$. Then for any $\gamma>0$,
	$$
	\tr_{\,\textup \Lp^2(\R^2,dx)} \left( H^\ppm + V \right)_\m^\gamma \leq q_\alpha^\gamma \, \tr_{\, \textup \Lp^2(\R^2,(1+|x|)^{\pm 2\alpha}\,dx)} \left( \mathcal H^\ppm - q_\alpha^{-1} (m^\pp m^\m)^2 V_\m \right)_\m^\gamma
	$$
	with the constants $q_\alpha$ from Proposition \ref{prop-lowerb} and $m^\ppm$ from \eqref{h-bounds}.
\end{corollary}

\begin{proof}
	According to Proposition \ref{prop-lowerb} we have for all $\phi\in C^1_c(\R^2)$
	\begin{align*}
		& Q^\ppm[e^{\pm h} \phi] + \int_{\R^2} V e^{\pm 2h} |\phi|^2 \,dx \\
		& \geq q_\alpha  (m^\mpp)^{-2} \left( \int_{\R^2} (1+|x|)^{\pm 2\alpha}\, |\nabla \phi|^2\, dx - q_\alpha^{-1} (m^\ppm m^\mpp)^2 \int_{\R^2} (1+|x|)^{\pm 2\alpha} V_\m |\phi|^2\,dx \right) 
	\end{align*}
	and
	$$
	\int_{\R^2} e^{\pm 2h} |\phi|^2 \,dx \geq (m^\mpp)^{-2} \int_{\R^2} (1+|x|)^{\pm 2\alpha} |\phi|^2\,dx \,.
	$$
	We know from \cite[Theorem 2.5]{ErVo} that the set $e^{\pm h} C^1_c(\R^2)$ is a form core for the operator $H^\ppm$. It is also easy to see that $C^1_c(\R^2)$ is a form core for $\mathcal H^\ppm$. Therefore these inequalities imply, by the variational principle
	$$
	N(H^\ppm + V+\tau) \leq N(\mathcal H^\ppm - q_\alpha^{-1} (m^\ppm m^\mpp)^2 V_\m + q_\alpha^{-1} \tau)
	\qquad\text{for all}\ \tau\geq 0 \,.
	$$
	Here, $N(T)$ denotes the number of negative eigenvalues, counting multiplicities, of a selfadjoint operator $T$. Using the identity
	$$
	\tr T_\m^\gamma = \gamma \int_0^\infty N(T+\tau)\, \tau^{\gamma-1}\,d\tau \,,
	$$
	we obtain the claimed inequality.
\end{proof}

\begin{remark}
	If instead of Proposition \ref{prop-lowerb} one uses inequality \eqref{eq:weightedupper}, one can argue similarly to prove the `reverse' inequality
	$$
	\tr_{\, \textup \Lp^2(\R^2,(1+|x|)^{\pm 2\alpha}\,dx)} \left( \mathcal H^\ppm + V \right)_\m^\gamma \leq 2^\gamma \, 
	\tr_{\, \textup \Lp^2(\R^2,dx)} \left( H^\ppm - 2^{-1} (m^\m m^\pp)^2 V_\m \right)_\m^\gamma
	$$ 
	In this sense the problem of proving Lieb--Thirring inequalities for $H^\ppm + V$ is equivalent, up to constants, to proving such inequalities for $\mathcal H^\ppm +V$. From now on we will deal with the latter problem.
\end{remark}

%%%%%%%%%%%%%%%%%%%%%%%%%%%%%%%%%%%%
%%%%%%%%%%%%%%%%%%%%%%%%%%%%%%%%%%%%

\section{Proof of Theorem \ref{thm-main-1}}\label{sec:subcritical}

In this section we prove the first one of our main results, Theorem \ref{thm-main-1}. This is substantially simpler than the second one, Theorem \ref{thm-main-3}, since either the operators are subcritical ($H^+$ with $\alpha>0$), or they are critical, but the endpoint value of $\gamma$ is excluded ($H^\ppm$ with $\alpha=0$).

%%%%%%%%%%%%%%%%%%%%%%%%%%%%%%%%%%%%
%%%%%%%%%%%%%%%%%%%%%%%%%%%%%%%%%%%%

\subsection{Proof of Theorem \ref{thm-main-1} for $\alpha=0$}

While the approach in the following subsection works for $\alpha=0$ as well, one can already at this point finish easily the proof in this case by adapting the argument in \cite{FrSiWe}.

\begin{proof}[Proof of Theorem \ref{thm-main-1} for $\alpha=0$]
	By the argument at the end of Subsection \ref{sec:strategy} we may assume $R=1$. For $\alpha=0$, the operators $\mathcal H^\ppm$ coincide with the Laplacian $-\Delta$ in $\Lp^2(\R^2)$. Therefore, Corollary \ref{comparison}, together with the usual Lieb--Thirring inequality in $\R^2$, see \eqref{lt-2D}, implies that for any $\gamma>0$
	$$
	\tr_{\, \Lp^2(\R^2,dx)} \left( H^\ppm + V \right)_\m^\gamma \leq \tr_{\, \Lp^2(\R^2,dx)} \left( -\Delta - (m^\pp m^\m)^2 V_\m \right)_\m^\gamma \leq L_\gamma (m^\pp m^\m)^{2(\gamma+1)} \int_{\R^2} V(x)_\m^{\gamma+1}\,dx \,.
	$$
	This is the claimed inequality.
\end{proof}

%%%%%%%%%%%%%%%%%%%%%%%%%%%%%%%%%%%%
%%%%%%%%%%%%%%%%%%%%%%%%%%%%%%%%%%%%

\subsection{Proof of Theorem \ref{thm-main-1} for general $\alpha$}

We will use of method of Lieb \cite{lieb} of proving Lieb--Thirring inequalities, which is based on a pointwise upper bound on the heat kernel. Such pointwise bounds have been studied in great generality by Grigor'yan, Saloff-Coste and others; see \cite{SC0,gsc,Gr,sc} and references therein. The usefulness of Grigor'yan--Saloff-Coste theory in the context of Lieb--Thirring inequalities was observed in~\cite{FrKo}.

In order to apply the results of Grigor'yan and Saloff-Coste it is convenient to exchange the weight $(1+|x|)^{\ppm 2\alpha}$ with the smooth weight $(1+|x|^2)^{\ppm\alpha}$. Strictly speaking, this replacement is not necessary, as one can verify that the relevant results of Grigor'yan--Saloff-Coste theory remain valid for our weight that this smooth away from a point and Lipschitz near that point. However, to shorten the presentation we will make this replacement at the expense of a further, controlled deterioration of the constant.

We consider the quadratic form
$$
\int_{\R^2} (1+|x|^2)^{\pm \alpha}\, |\nabla \phi|^2\, dx
$$
in the Hilbert space $\Lp^2(\R^2,(1+|x|^2)^{\pm \alpha}\,dx)$. This form, with form domain consisting of functions $\phi\in H^1_{\rm loc}(\R^2)\cap \Lp^2(\R^2,(1+|x|^2)^{\pm \alpha}\,dx)$ for which the form is finite, is nonnegative and closed. We denote the resulting selfadjoint, nonnegative operator in $\Lp^2(\R^2,(1+|x|^2)^{\pm \alpha}\,dx)$ by $\mathcal K^\ppm$.

Using the bounds
$$
2^{-\frac12}\, (1+|x|) \leq (1+|x|^2)^\frac12 \leq 1+|x|
$$
and proceeding as in the proof of Corollary \ref{comparison}, we find that
\begin{equation}\label{eq:comparison1}
	\begin{aligned}
		\tr_{\,\textup\Lp^2(\R^2,(1+|x|)^{\ppm 2\alpha}dx)} \left( \mathcal H^\ppm +V \right)_\m^\gamma & \leq \tr_{\,\textup\Lp^2(\R^2,(1+|x|^2)^{\ppm \alpha}dx)} \left( \mathcal K^\ppm - 2^{\alpha} \, V_\m \right)_\m^\gamma \,, \\[5pt]
		\tr_{\,\textup\Lp^2(\R^2,(1+|x|^2)^{\ppm \alpha}dx)} \left( \mathcal K^\ppm + V \right)_\m^\gamma
		& \leq 
		\tr_{\,\textup\Lp^2(\R^2,(1+|x|)^{\ppm 2\alpha}dx)} \left( \mathcal H^\ppm - 2^\alpha \, V_\m \right)_\m^\gamma \,.
	\end{aligned}
\end{equation}

In view of these inequalities we will now prove Lieb--Thirring inequalities for the operator $\mathcal K^\ppm$. Let
$$
p^\ppm(t; x,y) := e^{-t \mathcal K^\ppm}(x,y)
$$
denote the heat kernel generated by $\mathcal K^\ppm$. According to \cite[Equation (4.10)]{gsc}, for any $0\leq\alpha<1$ there is a constant $C$ such that
\begin{equation}  \label{p-upperb}
	p^\ppm (t; x,x) \, \leq \, C\  t^{-1} \, (1+|x| +\sqrt t)^{\mp 2\alpha} \qquad \forall\, t>0, \quad \forall \, x\in\R^2 \,.
\end{equation}

Let us comment on the bound \eqref{p-upperb}.  The result in \cite{gsc} is much more general. It gives matching upper and lower bounds for $p^\ppm (t; x,y)$ for general $x,y\in\R^2$. Also in the case of the upper sign the restriction $\alpha<1$ is not necessary. When comparing \eqref{p-upperb} with \cite[Equation (4.10)]{gsc}, note that our $\pm2\alpha$ plays the role of their $\alpha$. We also note that there is a typographical error in \cite[Equation (4.10)]{gsc}, which we have corrected in \eqref{p-upperb}. (Indeed, inserting the formula for $\mu_\alpha(B(x,r))$ before \cite[Equation (4.10)]{gsc} into \cite[Theorem 2.7]{gsc}, we see that $\alpha$ there needs to be replaced by $\alpha/2$.)

Lieb's method \cite{lieb} yields the upper bound 
\begin{equation} \label{lieb}
\tr\left(\mathcal K^\ppm + V\right)_\m^\gamma \, \leq \,
K_{a,\gamma}\, \int_{\R^2}\, \int_{0}^\infty  p^\ppm(t;x,x)\,
t^{-1-\gamma}\, (t\, V(x)+a)_\m\, dt\, (1+|x|^2)^{\ppm\alpha}\, dx \,,
\end{equation}
valid for any parameter $a>0$, with constant
\begin{equation}
K_{a,\gamma}= \Gamma(\gamma+1)\, \left(
e^{-a}-a\, \int_a^\infty\, s^{-1}\, e^{-s}\, ds\right)^{-1}.
\end{equation}
We now turn to the proof of our first main result.
 
\begin{proof}[Proof of Theorem \ref{thm-main-1}]
	Inequality \eqref{p-upperb} implies 
	\begin{equation}
		\label{eq:heatkernelsimpleplus}
		p^\pp(t; x,x) \, \leq \,  \frac{C_\alpha}{t} \, (1+|x|)^{-2\alpha} 
		\qquad \forall\, t>0, \quad \forall \, x\in\R^2 \, . 
	\end{equation}
	Inserting this into \eqref{lieb} we obtain, for any $\gamma>0$,
	$$
	\tr (\mathcal K^\pp + V)_\m^\gamma\   \leq\  C_\alpha\,  K_{a,\gamma}  \int_{\R^2}\, \int_{0}^\infty  
	t^{-2-\gamma}\, (t\, V(x)+a)_\m\, dt\, dx =  \frac{C_\alpha\,  K_{a,\gamma}}{a^\gamma\, \gamma(\gamma+1)}   \int_{\R^2} V(x)_\m^{\gamma+1} \, dx \,.
	$$
	Combining this with the bounds from Corollary \ref{comparison} and from \eqref{eq:comparison1} we obtain (for $R=1$, as we may assume) 
	\begin{align*}
		\tr_{\, \Lp^2(\R^2,dx)} \left( H^\pp + V \right)_\m^\gamma & \leq q_\alpha^\gamma \, \tr_{\, \Lp^2(\R^2,(1+|x|)^{2\alpha}\,dx)} \left( \mathcal H^\pp - q_\alpha^{-1} (m^\pp m^\m)^2 V_\m \right)_\m^\gamma \\
		& \leq q_\alpha^\gamma \, \tr_{\, \Lp^2(\R^2,(1+|x|^2)^{\alpha}\,dx)} \left( \mathcal K^\pp - q_\alpha^{-1} \, 2^\alpha \, (m^\pp m^\m)^2 V_\m \right)_\m^\gamma \\
		& = q_\alpha^{-1} \, 2^{\alpha(\gamma+1)} \, \frac{C_\alpha\,  K_{a,\gamma}}{a^\gamma\, \gamma(\gamma+1)} \,  (m^\pp m^\m)^{2(\gamma+1)} \, \int_{\R^2} V(x)_\m^{\gamma+1} \, dx \,,
	\end{align*}
	which is the claimed Lieb--Thirring inequality for $H^\pp+V$. The proof for $H^\m+V$ when $\alpha=0$ is similar.
\end{proof}

\begin{remark}\label{hminusineqality}
	It is interesting to note that the inequality in Theorem \ref{thm-main-3} can be obtained by the same method for $\gamma>\alpha$. Indeed, inequality \eqref{p-upperb} implies
	\begin{equation}
		\label{eq:heatkernelsimpleminus}
		p^\m(t; x,x) \, \leq\,  C_\alpha\, \big(  t^{-1} \, (1+|x|)^{2\alpha}  +  t^{\alpha-1 } \big) 
		\qquad \forall\, t>0, \quad \forall \, x\in\R^2 \,.
	\end{equation}
	Inserting this into \eqref{lieb} we obtain, for any $\gamma>\alpha$,
	\begin{align*}
		\tr (\h^\m + V)_\m^\gamma\  & \leq\  
		\frac{C_\alpha\,  K_{a,\gamma}}{a^\gamma\, \gamma(\gamma+1)}   \int_{\R^2} V(x)_\m^{\gamma+1}\,dx \\
		& \ \quad + \frac{C_\alpha \, K_{a,\gamma} }{a^{\gamma-\alpha}\, (\gamma-\alpha)(\gamma-\alpha+1)}  \int_{\R^2} (1+|x|)^{-2\alpha} \, V(x)_\m^{1+\gamma-\alpha}\, dx\,.
	\end{align*}
	Combining this with the bound from Corollary \ref{comparison}, we obtain Theorem \ref{thm-main-3} for $\gamma>\alpha$. Note that in the second term on the right side, we estimate
	$$
	(1+|x|)^{-2\alpha} \leq (m^\pp m^\m)^2 \, e^{-2(h(x)-h_0)} \,.
	$$
	(Indeed, $e^{-h(x)}\leq m^\m (1+|x|)^{-\alpha}$, so $e^{\|h\|_{\Lp^\infty(B(0,\epsilon))}}\leq m^\pp (1+\epsilon)^{\alpha}$, which according to our convention means that $e^{h_0}\leq m^\pp$.)		
\end{remark}

\begin{remark}\label{smallfieldrem}
	We claim that for any fixed $\gamma>0$ the limsup of the constants in the Lieb--Thirring inequalities in Theorems \ref{thm-main-1} and \ref{thm-main-3} remains finite as $\alpha\to 0$. This follows from the proofs that we have just given, together with the fact that the constants $C_\alpha$ in \eqref{eq:heatkernelsimpleplus} and \eqref{eq:heatkernelsimpleminus} remain bounded as $\alpha\to 0$. The latter claim follows from the explicit nature of the bounds in the Grigar'yan--Saloff-Coste theory. The basic ingredients, namely the volume doubling property and the Poincar\'e inequality (see \cite[Theorem 2.7]{gsc}), hold with constants that remain bounded as $\alpha\to 0$.
\end{remark}

%%%%%%%%%%%%%%%%%%%%%%%%%%%%%%%%%%%%
%%%%%%%%%%%%%%%%%%%%%%%%%%%%%%%%%%%%

\section{Proof of Theorem  \ref{thm-main-3}}\label{sec:critical}

In this section we prove the second of our main results, Theorem \ref{thm-main-3}. We will assume throughout that $0<\alpha<1$ and will prove this theorem only in the critical case $\gamma=\alpha$. This implies the result in the full regime $\gamma\geq\alpha$, either by the Aizenman--Lieb argument \cite{al} or by Remark \ref{hminusineqality}. Moreover, according to Corollary \ref{comparison} it suffices to prove the corresponding inequality for $\mathcal H^\m$ rather than $H^\m$.

%%%%%%%%%%%%%%%%%%%%%%%%%%%%%%%%%%%%
%%%%%%%%%%%%%%%%%%%%%%%%%%%%%%%%%%%%

\subsection{Reduction to radial functions}

For a function $f$ on $\R^2$ let
$$
\mathcal Pf(x) := (2\pi)^{-1} \int_0^{2\pi} f(|x|\cos\theta,|x|\sin\theta)\,d\theta
$$
and $\mathcal P^\bot := \1 - \mathcal P$. For any radial weight $w$ on $\R^2$, $\mathcal P$ is the orthogonal projection onto radial functions in $\Lp^2(\R^2,w(x)dx)$. The operator $\mathcal P$ commutes with $\mathcal H^\m$. Moreover, by the Schwarz inequality we have
$$
V_\m \leq 2 \mathcal P V_\m \mathcal P + 2 \mathcal P^\bot V_\m \mathcal P^\bot \,.
$$
From this, we conclude that
\begin{equation} \label{split-orth}
\tr(\mathcal H^\m + V)_\m^\gamma \leq \tr(\mathcal P(\mathcal H^\m - 2V_\m)\mathcal P)_\m^\gamma + \tr(\mathcal P^\bot(\mathcal H^\m - 2V_\m)\mathcal P^\bot )_\m^\gamma \,.
\end{equation}
We will treat the two terms on the right side separately. In this subsection we will treat the second term. We note that the first term, which will be treated in the remaining subsections, corresponds essentially to an operator in one dimension.

\begin{proposition} \label{prop-radial}
	For any $\gamma>0$,
	$$
	\tr(\mathcal P^\bot(\mathcal H^\m + V)\mathcal P^\bot )_\m^\gamma \leq \tfrac 98 \, L_\gamma \int_{\R^2} V(x)_\m^{\gamma+1}\,dx \,.
	$$ 
\end{proposition}

\begin{proof}
	We shall show that
	$$
	\tr_{\, \Lp^2(\R^2,(1+|x|)^{-2\alpha}dx)} (\mathcal P^\bot(\mathcal H^\m - V)\mathcal P^\bot )_\m^\gamma \leq \tr_{\, \Lp^2(\R^2,dx)} \Big(\mathcal P^\bot\big(-\tfrac 89 \Delta - V\big)\mathcal P^\bot \Big)_\m^\gamma \,,
	$$
	where we make the fact explicit that the traces on the two sides are in different Hilbert spaces. Once we have shown this inequality, we can appeal to the standard Lieb--Thirring inequality \eqref{lt-2D} to deduce the bound in the proposition.
	
	We consider the unitary operator $U: \Lp^2(\R^2,dx) \to \Lp^2(\R^2,(1+|x|)^{-2\alpha}dx)$, $\psi \mapsto (1+|x|)^{\alpha} \psi$. Since $U$ commutes with $\mathcal P$ and $V$, it suffices to show that
	$$
	U^* \mathcal P^\bot \mathcal H^\m \mathcal P^\bot U \geq \tfrac89 \,  \mathcal P^\bot(-\Delta)\mathcal P^\bot  \,.
	$$
	That is, we need to show that, if $\mathcal P\psi=0$, then
	$$
	\int_{\R^2} (1+|x|)^{-2\alpha} |\nabla ((1+|x|)^\alpha \psi)|^2\,dx \geq \tfrac89 \int_{\R^2} |\nabla \psi|^2\,dx \,.
	$$
	
	We compute
	\begin{align*}
		(1+|x|)^{-2\alpha} |\nabla ((1+|x|)^\alpha \psi)|^2
		& = (1+|x|)^{-2\alpha} |(1+|x|)^\alpha \nabla \psi + \alpha (1+|x|)^{\alpha-1} \tfrac x{|x|} \psi|^2 \\
		& = |\nabla \psi|^2 + \alpha^2 (1+|x|)^{-2} |\psi|^2 + \alpha (1+|x|)^{-1} \tfrac x{|x|}\cdot \nabla (|\psi|^2) \,.
	\end{align*}
	Integrating by parts, we obtain
	\begin{align*}
		\int_{\R^2} (1+|x|)^{-2\alpha} |\nabla ((1+|x|)^\alpha \psi)|^2\,dx
		& = \int_{\R^2} ( |\nabla \psi|^2 - (\alpha \nabla\cdot ( (1+|x|)^{-1} \tfrac x{|x|}) - \alpha^2(1+|x|)^{-2}) |\psi|^2 )\,dx \\
		& = \int_{\R^2} (|\nabla \psi|^2 - \alpha (1+|x|)^{-2} (|x|^{-1}-\alpha ) |\psi|^2 )\,dx \,.
	\end{align*}
	
	We introduce polar coordinate $x=(r\cos\theta,r\sin\theta)$. Since $\mathcal P^\bot(-\partial_\theta^2)\mathcal P^\bot\geq \mathcal P^\bot$, we have, if $\mathcal P\psi=0$,
	\begin{align*}
		\int_{\R^2} (1+|x|)^{-2\alpha} |\nabla ((1+|x|)^\alpha \psi)|^2\,dx
		& = \int_{\R^2} \left( (1+|x|)^{-2\alpha} |\partial_r (1+|x|)^\alpha \psi|^2 + |x|^{-2} |\partial_\theta \psi|^2 \right) dx \\
		& \geq \int_{\R^2} |x|^{-2} |\psi|^2\,dx \,.
	\end{align*}

	Combining the previous two equations we find, for any $\vartheta\in[0,1]$,
	\begin{align*}
		\int_{\R^2} (1+|x|)^{-2\alpha} |\nabla ((1+|x|)^\alpha \psi)|^2\,dx & \geq\,  \vartheta \int_{\R^2} |\nabla \psi|^2 \,dx \\
		& \quad + \int_{\R^2} ((1-\theta)|x|^{-2} - \vartheta \alpha (1+|x|)^{-2} (|x|^{-1}-\alpha )) |\psi|^2 \,dx \,.
	\end{align*}
	We can choose $\vartheta\in[0,1]$ (depending on $\alpha$) such that
	$$
	(1-\vartheta)|x|^{-2} - \vartheta \alpha (1+|x|)^{-2} (|x|^{-1}-\alpha )\geq 0
	\qquad \text{for all}\ x\in\R^2 \,.
	$$
	More precisely, we choose
	$$
	\vartheta := \left( \sup_{r>0} \left( 1+ \alpha (1+r)^{-2} r(1-\alpha r) \right) \right)^{-1}.
	$$
The supremum is attained at $r=1/(2\alpha+1)$, which leads to
	$$
	\vartheta = \frac{4(\alpha+1)}{5\alpha+4} \geq \frac 89 \,.
	$$
	This proves the claimed inequality.
\end{proof}

We note that the inequality $\alpha< 1$ that we assume throughout this section was only used at the very end of the previous proof when we bounded $\vartheta$ from below. Thus, an analogue of Proposition \ref{prop-radial} is valid even for $\alpha\geq 1$, but with a constant that depends on $\alpha$.

%%%%%%%%%%%%%%%%%%%%%%%%%%%%%%%%%%%%
%%%%%%%%%%%%%%%%%%%%%%%%%%%%%%%%%%%%

\subsection{Reduction to the lowest eigenvalue}\label{sec:redux2}

In the previous subsection we have treated the second term on the right side of \eqref{split-orth}. In this subsection we treat the first term, that is, we deal with the operator $\mathcal P(\mathcal H^\m+V)\mathcal P$. 

We let $\mathfrak h^\m$ denote the operator in $\Lp^2(\R_\pp,(1+r)^{-2\alpha}r\,dr)$ generated by the quadratic form
$$
\int_0^\infty (1+r)^{-2\alpha} |\phi'(r)|^2 r\,dr \,,
$$
defined on locally absolutely continuous functions $\phi$ on $\R_\pp$ belonging to $\Lp^2(\R_\pp,(1+r)^{-2\alpha}r\,dr)$ for which the integral is finite. If for a given function $V$ on $\R^2$ we let
\begin{equation}
	\label{eq:radialv}
	v(r) := \frac1{2\pi} \int_0^{2\pi} V(r\cos\theta,r\sin\theta)\,d\theta \,,
\end{equation}
then the nontrivial part of the operator $\mathcal P(\mathcal H^\m+V)\mathcal P$ is equal to $\mathfrak h^\m+v$ and, in particular,
\begin{equation}
	\label{eq:radial}
	\tr_{\, \Lp^2(\R^2,(1+|x|)^{-2\alpha}dx)} \left( \mathcal P(\mathcal H^\m+V)\mathcal P \right)_\m^\gamma = \tr_{\, \Lp^2(\R_\pp,(1+r)^{-2\alpha}r\,dr)} \left( \mathfrak h^\m +v \right)_\m^\gamma \,.
\end{equation}

In the remainder of this section we will treat $v$ as a given function on $\R_\pp$, ignoring that there is an underlying function $V$ on $\R^2$.

Our strategy to bound the right side of \eqref{eq:radial} will be to impose a Dirichlet boundary condition at $r=1$. This will result in two operators $\mathfrak h^\m_0$ and $\mathfrak h^\m_\infty$ in $\Lp^2((0,1),(1+r)^{-2\alpha}r\,dr)$ and $\Lp^2((1,\infty),(1+r)^{-2\alpha}r\,dr)$, respectively. These operators act in the same way as $\mathfrak h^-$, but functions in their form domain vanish at the point $r=1$. Since imposing a Dirichlet boundary condition is a rank one perturbation of the resolvent, it follows that
\begin{align} \label{dirichlet-1} 
\tr_{\, \Lp^2(\R_\pp,(1+r)^{-2\alpha}r\,dr)} \left( \mathfrak h^\m +v \right)_\m^\gamma \, 
&  \leq \, \tr_{\, \Lp^2((0,1),(1+r)^{-2\alpha}r\,dr)} \left( \mathfrak h_0^\m +v \right)_\m^\gamma \nonumber \\[4pt]
& \quad \, + \tr_{\, \Lp^2((1,\infty),(1+r)^{-2\alpha}r\,dr)} \left( \mathfrak h_\infty^\m +v \right)_\m^\gamma \nonumber\\[4pt]
 & \quad \, + \left( \inf\spec \left( \mathfrak h^\m +v \right) \right)_\m^\gamma \,.
\end{align}
In the following two propositions we will treat the first two terms on the right side, respectively. The third term will be treated in the next subsection.

\begin{proposition}\label{ltrad0}
	For any $\gamma>0$,
	$$
	\tr_{\, \textup \Lp^2((0,1),(1+r)^{-2\alpha}r\,dr)} \left( \mathfrak h_0^\m +v \right)_\m^\gamma
	\leq 2^{2\alpha(\gamma+1)} \, L_\gamma \, 2\pi \int_0^1 v(r)_\m^{\gamma+1}\,r\,dr \,.
	$$
\end{proposition}

\begin{proof}
	For functions $\phi$ in the form domain of $\mathfrak h_0^\m$ we bound
	\begin{align*}
		 \int_0^1 (1+r)^{-2\alpha} |\phi'(r)|^2 r\,dr + \int_0^1 v(r) |\phi(r)|^2 (1+r)^{-2\alpha} r\,dr  \geq 2^{-2\alpha} \int_0^1  |\phi'(r)|^2 r\,dr - \int_0^1 v(r)_\m |\phi(r)|^2 r\,dr
	\end{align*}
	and
	$$
	\int_0^1 (1+r)^{-2\alpha} |\phi(r)|^2 r\,dr \geq 2^{-2\alpha} 	\int_0^1 |\phi(r)|^2 r\,dr \,.
	$$
	By a similar argument as in the proof of Corollary \ref{comparison}, this implies
	$$
	\tr_{\, \Lp^2((0,1),(1+r)^{-2\alpha}r\,dr)} \left( \mathfrak h_0^\m +v \right)_\m^\gamma
	\leq \tr_{\, \Lp^2((0,1),r\,dr)} \left( -r^{-1} \partial_r r \partial_r - 2^{2\alpha} v_\m \right)_\m^\gamma \,,
	$$
	where the operator $-r^{-1} \partial_r r \partial_r  - 2^{2\alpha} v_\m$ is considered with a Dirichlet boundary condition at $r=1$. Note that this operator coincides with the nontrivial part of $\mathcal P(-\Delta - 2^{2\alpha} v(|\cdot|)_\m)\mathcal P$ acting in $\Lp^2(B(0,1),dx)$ with a Dirichlet boundary condition. Extending the operator to all of $\R^2$ and removing the projection $\mathcal P$ does not decrease the Riesz means and therefore we have, by the standard Lieb--Thirring inequality \eqref{lt-2D},
	\begin{align*}
		\tr_{\, \Lp^2((0,1),r\,dr)} \left( -r^{-1} \partial_r r \partial_r  - 2^{2\alpha} v_\m \right)_\m^\gamma
		& \leq \tr_{\, \Lp^2(\R^2,dx)} \left( -\Delta - 2^{2\alpha} \1_{B(0,1)}v(|\cdot|)_\m \right)_\m^\gamma \\
		& \leq 2^{2\alpha(\gamma+1)} \, L_\gamma \int_{B(0,1)} v(|x|)_\m^{\gamma+1}\,dx \\
		& = 2^{2\alpha(\gamma+1)} \, L_\gamma \, 2\pi \int_0^1 v(r)_\m^{\gamma+1}\,r\,dr \,.
	\end{align*}
	Combining this with the previous inequality yields the assertion.
\end{proof}

\begin{proposition}\label{ltradinfty}
	For any $\gamma>0$,
	$$
	\tr_{\, \textup \Lp^2((1,\infty),(1+r)^{-2\alpha}r\,dr)} \left( \mathfrak h_\infty^\m +v \right)_\m^\gamma
	\leq C_{\alpha,\gamma} \int_1^\infty v(r)_\m^{\gamma+1}\,r\,dr \,.
	$$
\end{proposition}

\begin{proof}
	Arguing similarly as at the beginning of the previous proof we find that
	$$
	\tr_{\, \Lp^2((1,\infty),(1+r)^{-2\alpha}r\,dr)} \left( \mathfrak h_0^\m +v \right)_\m^\gamma
	\leq \tr_{\, \Lp^2((1,\infty),r^{-2\alpha+1}\,dr)} \left( -r^{2\alpha-1} \partial_r r^{-2\alpha+1} \partial_r - 2^{2\alpha} v_\m \right)_\m^\gamma \,,
	$$
	where the operator $-r^{2\alpha-1} \partial_r r^{-2\alpha+1} \partial_r - 2^{2\alpha} v_\m$ is considered with a Dirichlet boundary condition at $r=1$. Extending the operator to all of $\R_\pp$ we will consider
	$$
	\tr_{\, \Lp^2(\R_\pp,r^{-2\alpha+1}\,dr)} \left( -r^{2\alpha-1} \partial_r r^{-2\alpha+1} \partial_r + \tilde v \right)_\m^\gamma \,,
	$$
	where $\tilde v = - 2^{2\alpha} v_\m$ on $(1,\infty)$ and $\tilde v=0$ on $(0,1)$. The operator $-r^{2\alpha-1} \partial_r r^{-2\alpha+1} \partial_r$ acts with a Dirichlet boundary condition at the origin. More precisely, it is defined as the closure of the quadratic form $\int_0^\infty r^{-2\alpha+1} |\phi'(r)|^2\,dr$ defined for $\phi\in C^1_c(\R_\pp)$. (We emphasize that $\R_\pp=(0,\infty)$, so functions in $C^1_c(\R_\pp)$ vanish in a neighborhood of the origin.)
	
	We consider the unitary operator $U:\Lp^2(\R_\pp,dr)\to \Lp^2(\R_\pp,r^{-2\alpha+1}dr)$, $\eta\mapsto r^{\alpha-\frac12} \eta$. Let us set $\phi(r) = \mathcal U \eta(r) = r^{\alpha-\frac12} \eta(r)$ and compute
	\begin{align*}
		|\phi'(r)|^2 & = \left| r^{\alpha-\frac12} \eta'(r) + (\alpha-\tfrac12) r^{\alpha-\frac32} \eta(r) \right|^2 \\
		& = r^{2\alpha-1} |\eta'(r)|^2 + 2(\alpha-\tfrac12) r^{2\alpha-2} \re \overline{\eta'(r)} \eta(r) + (\alpha-\tfrac12)^2 r^{2\alpha-3} |\eta(r)|^2 \\
		& = r^{2\alpha-1} |\eta'(r)|^2 + (\alpha-\tfrac12) r^{2\alpha-2} \left( |\eta(r)|^2 \right)' + (\alpha-\tfrac12)^2 r^{2\alpha-3} |\eta(r)|^2 \,, 
	\end{align*}
	leading to
	\begin{align*}
		\int_0^\infty |\phi'(r)|^2 r^{-2\alpha+1}\,dr & = \int_0^\infty \left( |\eta'(r)|^2 + (\alpha^2-\tfrac14) r^{-2} |\eta(r)|^2 \right) dr \,.
	\end{align*}
	Thus, we have shown that the operator $-r^{2\alpha-1}\partial_r r^{-2\alpha+1} \partial_r$ in $\Lp^2(\R_\pp,r^{-2\alpha+1}dr)$ is unitarily equivalent to the operator $-\partial_r^2 + (\alpha^2-\tfrac14) r^{-2}$ in $\Lp^2(\R_\pp,dr)$. It follows that
	$$
	\tr_{\, \Lp^2(\R_\pp,r^{-2\alpha+1}\,dr)} \left( -r^{2\alpha-1} \partial_r r^{-2\alpha+1} \partial_r + \tilde v \right)_\m^\gamma
	= \tr_{\, \Lp^2(\R_\pp\,dr)} \left( -\partial_r^2 + (\alpha^2-\tfrac14) r^{-2} + \tilde v \right)_\m^\gamma \,,
	$$

	For a lower bound we drop the term $\alpha^2 r^{-2}$ and recognize the operator $-\partial_r^2 -\tfrac14 r^{-2}$ as being unitarily equivalent to the radial part of the Laplace operator in $\R^2$. It follows that
	\begin{align*}
		\tr_{\, \Lp^2(\R_\pp\,dr)} \left( -\partial_r^2 + (\alpha^2-\tfrac14) r^{-2} + \tilde v \right)_\m^\gamma
		& \leq \tr_{\, \Lp^2(\R_\pp\,dr)} \left( -\partial_r^2 -\tfrac14 r^{-2} + \tilde v \right)_\m^\gamma \\
		& = \tr_{\, \Lp^2(\R^2,dx)}\left( \mathcal P(-\Delta +\tilde v(|\cdot|)) \mathcal P \right)_\m^\gamma \\
		& \leq \tr_{\, \Lp^2(\R^2,dx)}\left( -\Delta +\tilde v(|\cdot|) \right)_\m^\gamma \\
		& \leq L_\gamma \int_{\R^2} \tilde v(|x|)_\m^{\gamma+1}\,dx \\
		& = L_\gamma 2\pi \int_0^\infty \tilde v(r)_\m^{\gamma+1}\, r\,dr \,.
	\end{align*}
	Here we used the standard Lieb--Thirring inequality \eqref{lt-2D} on $\R^2$. Combining the previous inequalities yields the assertion.
\end{proof}

\begin{remark}
	There is an alternative way of finishing the proof without appealing to the Lieb--Thirring inequality \eqref{lt-2D}. Namely, when $\alpha\geq\frac12$ one can drop the term $(\alpha^2-\tfrac14) r^{-2}$ for a lower bound and for $0<\alpha<\frac12$ one can drop this term at the expense of reducing the constant 1 in front of $-\partial_r^2$ by Hardy's inequality. One arrives at having to bound $\tr_{\, \Lp^2(\R_\pp,dr)} (-\theta_\alpha \partial_r^2 + \tilde v )_\m^\gamma$ with  $\theta_\alpha:=\min\{ 1, 4\alpha^2\}$. This is possible in view of bounds by Egorov and Kondratiev \cite[Sec.~8.8]{ek}. 
	
	A drawback of the proof that we just sketched is that the constant diverges as $\alpha\to 0$ because of the presence of $\theta_\alpha$. This can be remedied by using more refined one-dimensional inequalities that take the Hardy term into account \cite{ef}. 
	
	In this connection it is interesting to note that for $0<\alpha\leq\frac12$ and $\gamma\geq \alpha$, the above proof also gives the bound
	$$
	\tr_{\, \Lp^2((1,\infty),(1+r)^{-2\alpha}r\,dr)} \left( \mathfrak h_\infty^\m +v \right)_\m^\gamma \leq \tilde C_{\alpha,\gamma} \int_1^\infty v(r)_\m^{1+\gamma-\alpha} r^{-2\alpha+1}\,dr \,.
	$$
	This follows from the fact, proved in \cite{ef}, that the inequality
	\begin{equation*}
		\tr_{\, \Lp^2(\R_\pp,dr)} (- \partial_r^2 - \tfrac14 r^{-2} + w )_\m^\gamma
		\leq L_{\gamma,a} \int_0^\infty w(r)_\m^{\gamma+\frac{1+a}2} r^a\,dr
	\end{equation*}
	is valid for $\gamma=\frac{1-a}2$ when $0\leq a<1$. We apply this inequality with $a=1-2\alpha$.
\end{remark}

%%%%%%%%%%%%%%%%%%%%%%%%%%%%%%%%%%%%
%%%%%%%%%%%%%%%%%%%%%%%%%%%%%%%%%%%%

\subsection{Bound on the lowest eigenvalue}

In the previous subsection we have bounded the first and second term on the right side of \eqref{dirichlet-1}. In this subsection we discuss the third term, that is, we discuss a lower bound on the lowest eigenvalue for $\mathfrak h^\m + v$. We shall prove the following bound

\begin{proposition}\label{prop-E-1}
	For any $0<\alpha<1$ there is a constant $C_\alpha$ such that
	$$
	\left( \inf\spec \left( \mathfrak h^\m + v \right) \right)_\m^\alpha 
	\leq C_\alpha \left( \int_0^\infty v(r)_\m^{1+\alpha} \, r \,dr + \int_0^\infty v(r)_\m \,  (1+r)^{-2\alpha}\, r\, dr \right).
	$$
\end{proposition}

It is natural to wonder whether in the bound in the proposition a single term on the right side suffices. This is not the case, as will be discussed in Appendix \ref{sec-one-term}.

Proposition \ref{prop-E-1} is in some sense the main step in the proof of our main result. It is certainly the most technical step and, indeed, in this subsection we only show how to reduce the proof to a technical lemma that will be verified in the following section. This lemma is stated in terms of the operator $T_\alpha$ in $\Lp^2(\R_\pp,dr)$ that is defined through the closure of the quadratic form 
\begin{equation} \label{T-form}
	\int_0^\infty |\eta'(r)|^2  \,dr -\alpha |\eta(1)|^2 -\frac 14  \int_0^1 \frac{ |\eta(r)|^2}{r^2}\, dr +  \Big(\alpha^2-\frac 14\Big) \int_1^\infty \frac{ |\eta(r)|^2}{r^2}\, dr
\end{equation}
defined for $\eta\in C^1_c(\R_\pp)$. We denote by $(T_\alpha+\kappa^2)^{-1}(r,r')$, $r,r'\in\R_\pp$, the integral kernel of the operator $(T_\alpha+\kappa^2)^{-1}$. In the following lemma we bound the difference between these kernels at $\alpha$ and at $0$.

\begin{lemma} \label{prop-hs}
	For any $\alpha \in (0,1)$ there is a constant $C_\alpha$ such that for all $\kappa>0$ and all $r,r'\in\R_\pp$ one has
	$$
	\left| (T_\alpha +\kappa^2)^{-1}(r,r') - (T_0 +\kappa^2)^{-1}(r,r') \right| \leq C_\alpha \, \kappa^{-2\alpha} \, \sqrt{rr'} \ (1+r)^{-\alpha} \, (1+r')^{-\alpha} \,.
	$$
\end{lemma}

Accepting this lemma for the moment, let us prove the main result of this subsection.

\begin{proof}[Proof of Proposition \ref{prop-E-1}]
	The proof will consist of two steps. In the first step we will prove the bound
	\begin{equation}
		\label{eq:lowesttalpha}
		\left( \inf\spec \left( T_\alpha + v \right) \right)_\m^\alpha 
		\leq C_\alpha' \left( \int_0^\infty v(r)_\m^{1+\alpha} \, r \,dr + \int_0^\infty v(r)_\m \,  (1+r)^{-2\alpha}\, r\, dr \right).
	\end{equation}
	and in a second step we will show that this inequality implies that in Proposition \ref{prop-E-1}.
	
	\emph{Step 1.} Let us denote 
	\begin{equation} \label{split} 
		G_0(\kappa) := (T_0 + \kappa^2)^{-1} \qquad \text{and} \qquad \Gamma_\alpha(\kappa) := (T_\alpha +\kappa^2)^{-1} -(T_0 + \kappa^2)^{-1} \,.
	\end{equation}
	By the variational principle, for the proof of \eqref{eq:lowesttalpha} we may assume that $v\leq 0$. We denote
	$$
	\kappa_* := \left( \inf\spec \left( T_\alpha + v \right) \right)_\m^\frac12 \,.
	$$
	We may assume that $\kappa_*>0$, for otherwise \eqref{eq:lowesttalpha} is trivially true.

	From the Birman--Schwinger  principle we deduce that 
	\begin{equation}  \label{bs-1}
		1 =  \left\| v_\m^\frac12 \, (T_\alpha+\kappa_*^2)^{-1} \, v_\m^\frac12 \right\| \ \leq \  \left\| v_\m^\frac12\, G_0(\kappa_*) \, v_\m^\frac12 \right\|  + \left\| v_\m^\frac12\, \Gamma_\alpha(\kappa_*) \, v_\m^\frac12 \right\| \,,
	\end{equation}
	where $\|\cdot\|$ denotes the operator norm in $\Lp^2(\R_\pp)$. We distinguish two cases depending on the size of the first term on the right side of \eqref{bs-1}.
	
	Assume first that $\left\| v_\m^\frac12\, G_0(\kappa_*) v_\m^\frac12 \right\|  \leq \frac 12$. Then, by \eqref{bs-1},
	\begin{equation}
		\label{eq:bs-2}
		\left\| v_\m^\frac12\, \Gamma_\alpha(\kappa_*) \, v_\m^\frac12 \right\| \geq \frac12 \,.
	\end{equation}
	Meanwhile, it follows from Lemma \ref{prop-hs} that
	$$
	\left\| v_\m^\frac12 \, \Gamma_\alpha(\kappa) \, v_\m^\frac12 \right\|_{\rm HS} \leq C_\alpha \kappa^{-2\alpha} \int_0^\infty v(r)_\m (1+r)^{-2\alpha} \, r\,dr \,,
	$$
	for all $\kappa>0$, 
	where $\|\cdot\|_{\rm HS}$ denotes the Hilbert--Schmidt norm in $\Lp^2(\R_\pp)$. Estimating the Hilbert--Schmidt norm from below by the operator norm and setting $\kappa=\kappa_*$, we obtain
	$$
	\left\| v_\m^\frac12 \, \Gamma_\alpha(\kappa_*) \, v_\m^\frac12 \right\| \leq C_\alpha \kappa_*^{-2\alpha} \int_0^\infty v(r)_\m (1+r)^{-2\alpha} \, r\,dr \,.
	$$
	Combining this with \eqref{eq:bs-2}, we obtain
	$$
	\kappa_*^{2\alpha} \leq 2 C_\alpha \int_0^\infty v(r)_\m (1+r)^{-2\alpha} \, r\,dr \,,
	$$
	which implies \eqref{eq:lowesttalpha}.
	
	If, on the contrary, $\left\| v_\m^\frac12\, G_0(\kappa_*) v_\m^\frac12 \right\| > \frac 12$, then the Birman-Schwinger principle implies that 
	\begin{equation*} 
		\inf \spec \big(T_0-2v_\m \big) \ < -\kappa_*^2 \,.
	\end{equation*} 
	Since $T_0$ is unitarily equivalent to the radial part of the Laplace operator in $\R^2$, cf.~\eqref{T-form}, so we infer that
	$$
	\inf\spec \big(-\Delta -2v(|\cdot|)_\m \big) \ < -\kappa_*^2 \,.
	$$
	Combining this with the usual Lieb--Thirring inequality \eqref{lt-2D}, we obtain
	$$
	\kappa_*^{2\alpha} \leq \tr \big(-\Delta -2v(|\cdot|)_\m \big)_\m^\alpha \leq 2^{\alpha+1} \, L_\alpha 2\pi \int_0^\infty v(r)_\m^{\alpha+1}\,r\,dr \,,
	$$
	which implies \eqref{eq:lowesttalpha}. (Note that instead of the Lieb--Thirring inequality \eqref{lt-2D} the so-called one-particle Lieb--Thirring inequality, that is, a Sobolev interpolation inequality \cite[Subsection 5.1.2]{flw-book}, would suffice.) This completes the proof of \eqref{eq:lowesttalpha}.
	
	\medskip
	
	\emph{Step 2.} We now deduce the bound in the proposition from the bound \eqref{eq:lowesttalpha}. This is achieved by bringing the weight $(1+r)^{-2\alpha} r$ appearing for the operator $\mathfrak h^\m$ into a more canonical form and then applying a unitary transformation to remove this more canonical weight. We define
	\begin{equation} \label{Weight}
		w(r) := \begin{cases}
			r  & \text{if }\quad   0 < r\leq 1 \,,
			\\[5pt] 
			r^{1-2\alpha} & \text{if }\quad  1 < r < \infty \,.
		\end{cases}
	\end{equation}
	We denote by $h_\alpha$ the operator in $\Lp^2(\R_\pp, w(r) dr)$ associated with the quadratic form 
	\begin{equation}
		\int_0^\infty |\phi'(r)|^2\,  w(r) \,dr \,,
	\end{equation}
	defined on functions $\phi\in \Lp^2(\R_\pp, w(r) dr)$ that are locally absolutely continuous on $\R_\pp$ and for which the quadratic form is finite. Using the bounds
	$$
	2^{-2\alpha} \, w(r) \leq (1+r)^{-2\alpha} r \leq w(r) \,,
	$$
	we find, similarly as in the proof of Corollary \ref{comparison},
	$$
	\left( \inf\spec_{\Lp^2(\R_\pp, (1+r)^{-2\alpha}r dr)} \left( \mathfrak h^\m + v \right) \right)_\m
	\leq \left( \inf\spec_{\Lp^2(\R_\pp, w(r) dr)} \left( h_\alpha - 2^{2\alpha} v_\m \right) \right)_\m \,.
	$$
	This reduces the proof of the bound in the proposition to the proof of the bound for $h_\alpha - 2^{2\alpha} v_\m$.
	
	The unitary mapping $\mathcal U :\Lp^2(\R_\pp,dr)\to \Lp^2(\R_\pp, w(r) dr)$ given by $\mathcal U \eta(r) := w(r)^{-\frac12} \eta(r)$ satisfies 
	$$
	\int_0^\infty |(\mathcal U\eta)'|^2 w(r)\,dr = \int_0^\infty |\eta'(r)|^2  \,dr -\alpha |\eta(1)|^2 -\frac 14  \int_0^1 \frac{ |\eta(r)|^2}{r^2}\, dr +  \Big(\alpha^2-\frac 14\Big) \int_1^\infty \frac{ |\eta(r)|^2}{r^2}\, dr \,.
	$$
	Clearly, the form core $C^1_c(\R_\pp)$ of $T_\alpha$ is mapped into the form domain of $h_\alpha$. Conversely, arguing as in \cite[Lemma 2.33]{flw-book} one can show that $C^1_c(\R_\pp)$ is a form core of $h_\alpha$, and the image of it under $\mathcal U^{-1}$ is in the form domain of $T_\alpha$. These facts imply that
	$$
	\mathcal U^{-1}\, h_\alpha\,  \mathcal U = T_\alpha \,.
	$$
	As a consequence, for any $\tilde v$ (in particular, for $\tilde v=-2^{2\alpha}v_\m$)
	$$
	\inf\spec_{\Lp^2(\R_\pp, w(r) dr)} \left( h_\alpha +\tilde v \right)
	= \inf\spec_{\Lp^2(\R_\pp, dr)} \left( T_\alpha +\tilde v \right) \,.
	$$
	This concludes the proof of the proposition.	 
\end{proof}

%%%%%%%%%%%%%%%%%%%%%%%%%%%%%%%%%%%
%%%%%%%%%%%%%%%%%%%%%%%%%%%%%%%%%%%

\subsection{Proof of Theorem \ref{thm-main-3}}

We finally put all the ingredients from this section together and prove our second main result.

\begin{proof}[\bf Proof of Theorem \ref{thm-main-3}]
	Given a sufficiently regular real function $V$ on $\R^2$ we define the function $v$ on $\R_\pp$ by \eqref{eq:radialv}. Combining \eqref{eq:radial}, \eqref{dirichlet-1} and Propositions \ref{ltrad0}, \ref{ltradinfty} and \ref{prop-E-1}, we see that for each $0<\alpha<1$ there is a constant $C_\alpha$ such that
	\begin{align*}
		\tr_{\, \Lp^2(\R^2,(1+|x|)^{-2\alpha}dx)} \left( \mathcal P(\mathcal H^\m+V)\mathcal P \right)_\m^\alpha
		\leq C_\alpha \left( \int_0^\infty v(r)_\m^{1+\alpha} \, r \,dr + \int_0^\infty v(r)_\m \,  (1+r)^{-2\alpha}\, r\, dr \right) \\
		\leq \frac{C_\alpha}{2\pi} \left( \int_{\R^2} V(x)_\m^{1+\alpha} \,dx + \int_{\R^2} V(x)_\m \,  (1+|x|)^{-2\alpha}\, dx \right).
	\end{align*}
	The last inequality comes from H\"older's inequality for the angular integration. Combining this inequality with \eqref{split-orth} and Proposition \ref{prop-radial}, we see that for each $0<\alpha<1$ there is a constant $C_\alpha'$ such that
	\begin{align*}
		\tr_{\, \Lp^2(\R^2,(1+|x|)^{-2\alpha}dx)} \left( \mathcal H^\m+V \right)_\m^\alpha
		\leq C_\alpha' \left( \int_{\R^2} V(x)_\m^{1+\alpha} \,dx + \int_{\R^2} V(x)_\m \,  (1+|x|)^{-2\alpha}\, dx \right).
	\end{align*}
	Under Assumption \ref{ass} with $R=1$ (and, as everywhere in this section, $\alpha>0$), we can use the same argument as in Remark \ref{hminusineqality} to replace the weight $(1+|x|)^{-2\alpha}$ in the second term by $(m^\pp m^\m)^2 \, e^{-2(h(x)-h_0)}$. The claimed inequality in Theorem \ref{thm-main-3} for $\gamma=\alpha$ then follows from Corollary \ref{comparison}. Note that this yields, in particular, the claimed dependence of the constants on the product $m^\pp m^\m$. As we have already mentioned at the beginning of this section, the claimed inequality for $\gamma>\alpha$ follows either from the inequality for $\gamma=\alpha$ by the Aizenman--Lieb argument \cite{al} or by Remark \ref{hminusineqality}. Finally, the case $R\neq 1$ can be reduce to the case $R=1$ by scaling as we already observed.
\end{proof}

%%%%%%%%%%%%%%%%%%%%%%%%%%%%%%%%%%%
%%%%%%%%%%%%%%%%%%%%%%%%%%%%%%%%%%%

\section{Proof of Lemma \ref{prop-hs}}

In this section we will give the proof of Lemma \ref{prop-hs}. We use the notation $\Gamma_\alpha(\kappa)$ from \eqref{split} for the resolvent difference. Thus, we are looking for a pointwise bound on the integral kernel of the operator $\Gamma_\alpha(\kappa)$. 

%%%%%%%%%%%%%%%%%%%%%%%%%%%%%%%%%%%%
%%%%%%%%%%%%%%%%%%%%%%%%%%%%%%%%%%%%

Our first goal is to find an explicit formula for this integral kernel. Let
\begin{equation}\label{coef}
	\begin{aligned} 
		A_\alpha(\kappa) & := \kappa I_0(\kappa) K_{\alpha+1}(\kappa) +\kappa I_1(\kappa) K_{\alpha}(\kappa) -2\alpha I_0(\kappa) K_{\alpha}(\kappa) \,, \\[2pt]
		B_\alpha(\kappa) & := \kappa I_0(\kappa) I_{\alpha+1}(\kappa) -\kappa I_1(\kappa) I_{\alpha}(\kappa) +2\alpha I_0(\kappa) I_{\alpha}(\kappa) \,, \\[2pt]
		D_\alpha(\kappa) & := \kappa K_1(\kappa) K_{\alpha}(\kappa) -\kappa K_0(\kappa) K_{\alpha+1}(\kappa) +2\alpha K_0(\kappa) K_{\alpha}(\kappa) \,,
	\end{aligned}
\end{equation}
and put
$$
f_\alpha(\kappa) := \frac{D_\alpha(\kappa) }{A_\alpha(\kappa)} \qquad \text{and}\qquad 
g_\alpha(\kappa) := \frac{B_\alpha(\kappa) }{A_\alpha(\kappa)} \, .
$$
It will turn out that $A_\alpha(\kappa)\neq 0$, so $f_\alpha(\kappa)$ and $g_\alpha(\kappa)$ are well defined.

\begin{lemma}
	\label{intkernel}
	One has $A_\alpha(\kappa)\neq 0$ for all $\kappa>0$. Moreover, for all $\kappa>0$, $\alpha\in(0,1]$ and $0<r\leq r'<\infty$,
	\begin{align} 
		\Gamma_\alpha(r,r';\kappa) = \sqrt{r r'}\, \times
		\begin{cases}
			f_\alpha(\kappa)  I_0(\kappa r) I_0(\kappa r')  & \text{if}\  0 < r\leq r'\leq  1 \,,
			\\[5pt] 
			K_\alpha(\kappa r') I_\alpha(\kappa r)-  I_0(\kappa r) K_0(\kappa r') + g_\alpha(\kappa)  K_\alpha(\kappa r) K_\alpha(\kappa r')  & \text{if}\ 1 < r \leq r' \,, \label{G-alpha} \\[5pt]
			I_0(\kappa r)\, \big( A^{-1}_\alpha(\kappa) \, K_\alpha(\kappa r')  -K_0(\kappa r')\big) & \text{if}  \  0 < r \leq 1 \leq  r' \,.
		\end{cases}
	\end{align}
	The same formula is valid when $r>r'$, provided the variables $r$ and $r'$ are interchanged.
\end{lemma}

\begin{proof}
	We begin by deriving a formula for the integral kernel of $(T_\alpha+\kappa^2)^{-1}$ for $\alpha\in[0,1]$. By Sturm--Liouville theory it can be written in terms of two solutions $v_1$ and $v_2$ of the system
	\begin{align*}
		-v'' - \tfrac14 \, r^{-2} v & = -\kappa^2 v \qquad \text{in}\ (0,1) \,, \\
		-v'' + (\alpha^2-\tfrac14) \, r^{-2} v & = -\kappa^2 v \qquad \text{in}\ (1,\infty) \,, \\
		v(1_-) & = v(1_+) \\
		v'(1_-) & = v'(1_+) + \alpha v(1_+) \,.
	\end{align*}
	(The jump condition at $r=1$ comes from the term $\alpha|\eta(1)|^2$ in the quadratic form \eqref{T-form} of the operator $T_\alpha$.) The solution $v_1$ is supposed to lie in the form domain of $T_\alpha$ near the origin and the solution $v_2$ is supposed to be square-integrable at infinity.
	
	Using standard facts about Bessel's equation \cite[Sec.~9]{as}, we find that these two solutions are given by 
	\begin{align} 
		v_1(r) & = \sqrt{r} \times \begin{cases}
			I_0(\kappa r)  & \text{if }\quad   0 < r\leq 1 \,,
			\\[5pt] 
			\widetilde A_\alpha(\kappa) I_\alpha(\kappa r) + \widetilde B_\alpha(\kappa) K_\alpha(\kappa r) & \text{if }\quad  1 < r < \infty \,,
		\end{cases}
	\end{align}
	and 
	\begin{align} 
		v_2 (r) & = \sqrt{r} \times \begin{cases}
			\widetilde D_\alpha(\kappa) I_0(\kappa r) +  \widetilde C_\alpha(\kappa) K_0(\kappa r)  & \text{if }\quad   0 < r\leq 1 \,,
			\\[5pt] 
			K_\alpha(\kappa r) & \text{if }\quad  1 < r < \infty \,,
		\end{cases}
	\end{align}
	with coefficients $\widetilde A_\alpha(\kappa), \widetilde B_\alpha(\kappa), \widetilde C_\alpha(\kappa)$ and $\widetilde D_\alpha(\kappa)$ that are determined by the continuity and jump conditions at $r=1$. (We will give explicit expressions later in this proof.) Using the Wronski relation \cite[Eq.~9.6.15]{as} for the Bessel functions, viz.
	\begin{equation}\label{wronsk}
		W\big\{ K_\nu(z), I_\nu(z)\big\} = I _\nu(z) K_{\nu+1} (z) +K_\nu(z) I_{\nu+1} (z) = \frac 1z \,,
	\end{equation}
	we obtain
	$$
	W\big\{ v_1, v_2 \big\} = \widetilde C_\alpha(\kappa) = \widetilde A_\alpha(\kappa) \, .
	$$
	
	Let us show that $\widetilde A_\alpha(\kappa)\neq 0$ for all $\kappa>0$. Indeed, if we had $\widetilde A_\alpha(\kappa_0)=0$ for some $\kappa_0>0$, then $v_1$ would be an eigenfunction of $T_\alpha$ with eigenvalue $-\kappa_0^2$, but this contradicts the fact that $T_\alpha$ is a nonnegative operator. The latter fact follows from Step 2 in the proof of Proposition \ref{prop-E-1}, where we showed that $T_\alpha$ is unitarily equivalent to the manifestly nonnegative operator $h_\alpha$.
	
	By Sturm--Liouville theory, it follows from the above facts that for any $\kappa>0$, the integral kernel of $(T_\alpha+\kappa^2)^{-1}$ is given by
	\begin{equation}\label{green}
		\begin{aligned} 
			(T_\alpha+\kappa^2)^{-1}(r,r')= \sqrt{r r'} \times
			\begin{cases}
				K_0(\kappa r') I_0(\kappa r) + \widetilde f_\alpha(\kappa)  I_0(\kappa r) I_0(\kappa r')  & \text{if}\  0 < r\leq r'\leq  1 \,,
				\\[5pt] 
				K_\alpha(\kappa r') I_\alpha(\kappa r) + \widetilde g_\alpha(\kappa)  K_\alpha(\kappa r) K_\alpha(\kappa r')  & \text{if}\    1 < r \leq r' \,, \\[5pt]
				\widetilde A^{-1}_\alpha(\kappa)\, I_0(\kappa r)  K_\alpha(\kappa r')  & \text{if}\    0 < r \leq 1 \leq r' \,,
			\end{cases}
		\end{aligned}
	\end{equation}
	where we have denoted 
	$$
	\widetilde f_\alpha(\kappa) := \frac{\widetilde D_\alpha(\kappa) }{\widetilde A_\alpha(\kappa)} \qquad \text{and}\qquad 
	\widetilde g_\alpha(\kappa) := \frac{\widetilde B_\alpha(\kappa) }{\widetilde A_\alpha(\kappa)} \, .
	$$
	As usual, the formula for $r> r'$ follows by interchanging the variables. 
	
	Note also that $\widetilde A_0(\kappa)=\widetilde C_0(\kappa)=1$ and $\widetilde D_0(\kappa)=\widetilde B_0(\kappa)=0$, so, in particular, 
	$$
	\widetilde f_0(\kappa)=\widetilde g_0(\kappa)=0 \,.
	$$
	Thus, recalling the definition \eqref{split}, we see that \eqref{green} implies the formula in the lemma, except that the untilded quantities appear there rather than the tilded ones. Thus, to complete the proof we need to show that the former coincide with the latter. To do so, we replace the derivatives of Bessel functions appearing in the jump condition at $r=1$ using \cite[Eq.~9.6.26]{as} in terms of Bessel functions without derivatives. Solving the corresponding system of four linear equations with four unknown, we see that $\widetilde A_\alpha(\kappa), \widetilde B_\alpha(\kappa), \widetilde D_\alpha(\kappa)$ are given by the expressions on the right side of \eqref{coef}, as claimed. (Note that we have already shown that $\widetilde C_\alpha(\kappa) = \widetilde A_\alpha(\kappa)$, which is confirmed by the solution of the system of linear equations.)
\end{proof}

Next, we bound the quantities appearing in Lemma \ref{intkernel}.

\begin{lemma}\label{constantbounds}
	Let $\alpha\in(0,1]$. The following bounds hold for all $\kappa>0$ with an implicit constant depending possibly on $\alpha$:
	\begin{align*}
		|A_\alpha(\kappa)^{-1} - 1| & \lesssim \kappa^{-2\alpha} \1_{(0,1)}(\kappa) + \kappa^{-1} \1_{[1,\infty)}(\kappa) \,,\\
		|f_\alpha(\kappa)| & \lesssim \kappa^{-2\alpha} \1_{(0,1)}(\kappa) + \kappa^{-1} e^{-2\kappa} \1_{[1,\infty)}(\kappa) \,,\\
		|g_\alpha(\kappa)| & \lesssim \1_{(0,1)}(\kappa) + \kappa^{-1} e^{2\kappa} \1_{[1,\infty)}(\kappa) \,. 
	\end{align*} 
\end{lemma}

In fact, in the following proof we will establish the asymptotic behavior of the three quantities in the lemma for $\kappa\to 0$ and $\kappa\to\infty$. The above bounds, however, are all that we need.

For the proof we make use of the following asymptotic facts about Bessel functions:
\begin{equation}  \label{KI-zero}
	\begin{aligned}
		K_\nu(z) & = 
		\begin{cases}
			-\ln z + C + \mathcal O(z^2 |\ln z|) & \text{if}\ \nu=0 \,,\\
			\left( \frac z2 \right)^{-\nu} \tfrac12\, \Gamma(\nu) - \left( \frac z2 \right)^\nu \frac1{2\nu}\,\Gamma(1-\nu) + \mathcal O(z^{2-\nu})		& \text{if}\ 0<\nu<1 \,, \\
			z^{-1} + \mathcal O(z|\ln z|) & \text{if}\ \nu=1 \,, \\
			\left( \frac z2 \right)^{-\nu} \tfrac12\, \Gamma(\nu) + \mathcal O(z^{2-\nu})		& \text{if}\ \nu>1 \,,
		\end{cases}
		\qquad\text{as}\ z\to 0 \,, \\[6pt]
		I_\nu(z) & = \left( \frac z2 \right)^\nu \Gamma(\nu+1)^{-1} + \mathcal O(z^{2+\nu}) \qquad\text{as}\ z\to 0 \,,
	\end{aligned}
\end{equation} 
see \cite[(9.6.10) and (9.6.2)]{as}, and
\begin{equation}  \label{KI-infty}
	\begin{aligned}
		K_\nu(z) & = \sqrt{\frac{\pi}{2z}}\ e^{-z}\Big(1 +\frac{4\nu^2-1}{8z} + \mathcal{O}(z^{-2})\Big) \qquad\text{as}\ z\to\infty \,, \\[6pt]
		I_\nu(z) & = \sqrt{\frac{1}{2 \pi z}}\ e^{z}\Big(1 -\frac{4\nu^2-1}{8z} + \mathcal{O}(z^{-2})\Big) \qquad\text{as}\ z\to\infty \,.
	\end{aligned}
\end{equation} 
see \cite[(9.7.1) and (9.7.2)]{as}. The constant in \eqref{KI-zero} for $\nu=0$ is known, but its value is irrelevant for our purposes.

In addition, we use the following global properties of Bessel functions: $K_\nu$ and $I_\nu$ are positive. Moreover, $I_0$ is decreasing. (In fact, $I_\nu$ is decreasing and $K_\nu$ is increasing for any $\nu\geq 0$, but we will not need this.)

\begin{proof}
	\emph{Step 1. Asymptotics at the origin.} From \eqref{KI-zero} we deduce that, as $\kappa\to 0$,
	\begin{align*}
		A_\alpha(\kappa) &= 2^{-\alpha} \Gamma(1-\alpha) \,\kappa^\alpha + \mathcal O(\kappa^{2-\alpha}) \,,\\[4pt]
		B_\alpha(\kappa) &= 2^{1-\alpha} \Gamma(\alpha)^{-1} \, \kappa^\alpha + \mathcal{O}(\kappa^{2+\alpha}) \,,\\[4pt]
		D_\alpha(\kappa) &= 2^{-1-\alpha} \Gamma(\alpha) \, \kappa^{-\alpha} + \mathcal O(\kappa^\alpha \ln\kappa) \,. 
	\end{align*}
	We note in this computation there is a cancellation at order $\kappa^{-\alpha}$ for $A_\alpha(\kappa)$ and at order $\kappa^{-\alpha}\ln\kappa$ for $D_\alpha(\kappa)$.
	
	These asymptotics imply that, as $\kappa\to 0$,
	\begin{equation*} %\label{fg-asymp-small}
		\begin{aligned}
			f_\alpha(\kappa)\ &  = \frac{\Gamma(\alpha)}{2\,\Gamma(1-\alpha)}\, \kappa^{-2\alpha} + \mathcal{O}(|\ln\kappa| + \kappa^{2-4\alpha}) \,,  \\[4pt]
			g_\alpha(\kappa)\ &  = \frac{2}{\Gamma(\alpha)\,\Gamma(1-\alpha)} +  \mathcal{O}(\kappa^{2-2\alpha}) \,.
		\end{aligned}
	\end{equation*}
	
	\medskip
	
	\emph{Step 2. Asymptotics at infinity.} From \eqref{KI-infty} we deduce that, as $\kappa\to\infty$,
	\begin{equation*} %\label{ABD-k-large}
		\begin{aligned}
			A_\alpha(\kappa) &= 1+   \mathcal{O}(\kappa^{-1}) \,,  \\[3pt]
			B_\alpha(\kappa) &= \frac{\alpha}{2\pi \kappa}\, e^{2\kappa}\, \big(1+  \mathcal{O}(\kappa^{-1}) \big) \,,  \\[3pt]
			D_\alpha(\kappa) &=  \frac{\alpha \pi}{2 \kappa} \, e^{-2\kappa} \, \big(1+  \mathcal{O}(\kappa^{-1}) \big) \,.
		\end{aligned}
	\end{equation*}
	These asymptotics imply that, as $\kappa\to\infty$,
	\begin{equation*} %\label{fg-asymp-large}
		\begin{aligned}
			f_\alpha(\kappa)\ &  = \frac{\alpha\pi}{2\kappa}\, e^{-2\kappa}  \big(1+  \mathcal{O}(\kappa^{-1}) \big) \,, \\[5pt]
			g_\alpha(\kappa)\ &  =  \frac{\alpha}{2\pi \kappa}\, e^{2\kappa} \big(1+  \mathcal{O}(\kappa^{-1}) \big) \,.
		\end{aligned}
	\end{equation*}
	
	\medskip
	
	\emph{Step 3. Uniform bounds.}
	The asymptotics in Steps 1 and 2 imply that the claimed bounds in the lemma hold for all sufficiently small and all sufficiently large $\kappa$. Since $A_\alpha(\kappa)$, $B_\alpha(\kappa)$ and $D_\alpha(\kappa)$ are continuous functions of $\kappa$ and since $A_\alpha(\kappa)$ does not vanish according to Lemma \ref{intkernel}, we obtain the claimed bounds for all $\kappa>0$.
\end{proof}

After these preparations we are ready to prove the claimed pointwise bound on the the integral kernel of $\Gamma_\alpha(\kappa)$.

\begin{proof}[Proof of Lemma \ref{prop-hs}]
	By selfadjointness, it suffices to prove the bound for $r\leq r'$, which we will assume throughout the proof. We split the integral kernel $\Gamma_\alpha(r,r';\kappa)$ of $\Gamma_\alpha(\kappa)$ as
	\begin{equation} \label{R-eq}
		\Gamma_\alpha(r,r';\kappa) = \sqrt{r r'} \left( \mathcal{S}(r,r';\kappa)+\rr(r,r';\kappa) \right)
	\end{equation}
	with 
	\begin{align} 
		\s(r,r';\kappa) = 
		\begin{cases}
			f_\alpha(\kappa)  I_0(\kappa r) I_0(\kappa r')  & \text{if}\  0 < r\leq r'\leq  1 \,,
			\\[5pt] 
			g_\alpha(\kappa)  K_\alpha(\kappa r) K_\alpha(\kappa r')  & \text{if}\  1 < r \leq r' \,, \\[5pt]
			I_0(\kappa r) \big(A^{-1}_\alpha(\kappa)  K_\alpha(\kappa r') - K_0(\kappa r') \big) & \text{if}\  0 < r \leq 1 \leq  r' \,,  \label{S-def}
		\end{cases} 
	\end{align}
	and 
	\begin{align}  \label{def-rr}
		\rr(r,r';\kappa) =
		\begin{cases}
			K_\alpha(\kappa r') I_\alpha(\kappa r)-  I_0(\kappa r) K_0(\kappa r')  & \text{if}\  1 < r \leq r' \,, \\[5pt]
			0  & {\rm  elsewhere} \,,
		\end{cases}
	\end{align}
	and show that both pieces satisfy the bound claimed in the lemma. This is the content of the following two respective steps.
	
	To simplify the notation, we shall use the symbol $\lesssim$ to indicate the existence of a constant such that the inequality holds when the right side is multiplied by this constant. The constant may depend on $\alpha\in(0,1]$, but is independent of $0<r\leq r'<\infty$ and $\kappa>0$.
		
	\bigskip
	
	\emph{Step 1.} In this step we show that $|\s(r,r';\kappa)|\lesssim \kappa^{-2\alpha} \, (1+r)^{-\alpha} (1+r')^{-\alpha}$ for all $0<r\leq r'<\infty$.
	
	We distinguish three cases. 
	
	First, let $0<r \leq r'\leq 1$. We claim that
	\begin{equation} \label{sup-I}
		\sup_{\kappa >0,\, \rho \leq1}  \kappa ^{2\alpha} | f_\alpha(\kappa)| I^2_0(\kappa \rho) <\infty \,.
	\end{equation}
	Once we have shown this, we deduce that
	$$
	|\s(r,r';\kappa)| =  | f_\alpha(\kappa) |  I_0(\kappa r) I_0(\kappa r')  \  \lesssim \kappa^{-2\alpha} \lesssim \kappa^{-2\alpha} (1+r)^{-\alpha} (1+r')^{-\alpha} \,,
	$$
	which is the claimed bound.
	
	To prove \eqref{sup-I}, we first assume $\kappa\leq 1$. Then, by \eqref{KI-zero}, $I_0(\kappa\rho)^2\lesssim 1$ for all $\rho\leq 1$ and, by Lemma \ref{constantbounds}, $|f_\alpha(\kappa)|\lesssim \kappa^{-2\alpha}$. This proves \eqref{sup-I} for $\kappa\leq 1$. Now let $\kappa\geq 1$. Then, by \eqref{KI-infty} and the monotonicity of $I_0$, $I_0(\kappa\rho)^2 \leq I_0(\kappa)^2 \lesssim \kappa^{-1} e^{2\kappa} \lesssim \kappa^{-2\alpha+1} e^{2\kappa}$ for all $\rho\leq 1$. Moreover, by Lemma \ref{constantbounds}, $|f_\alpha(\kappa)|\lesssim \kappa^{-1} e^{-2\kappa}$. This proves \eqref{sup-I} for $\kappa\geq 1$.	
	
	Next, let $1 \leq r \leq r'$. We claim that
	\begin{equation} \label{sup-K}
		\sup_{\kappa >0,\, \rho\geq 1}  (\kappa \rho)^{2\alpha} | g_\alpha(\kappa)| K^2_\alpha(\kappa \rho) < \infty \,.
	\end{equation}
	Once we have shown this, we deduce that
	$$ 
		|\s(r,r';\kappa)| = | g_\alpha(\kappa) |  K_\alpha(\kappa r) K_\alpha(\kappa r')  \  \lesssim (\kappa r)^{-2\alpha} (\kappa r')^{-2\alpha} \lesssim 		 \kappa^{-2\alpha} (1+r)^{-\alpha} (1+r')^{-\alpha}\,,
	$$
	which is the claimed bound.
	
	To prove \eqref{sup-K}, we first assume $\kappa\leq 1$. Then, by Lemma \ref{constantbounds}, $|g_\alpha(\kappa)|\lesssim 1$. Meanwhile, it follows from \eqref{KI-zero} and \eqref{KI-infty} that
	\begin{equation*}
		%\label{eq:besselkglobal}
		\sup_{z>0} z^\alpha K_\alpha(z) <\infty \,.
	\end{equation*}
	(We emphasize that $\alpha>0$.) This proves \eqref{sup-K} for $\kappa\leq 1$. Now let $\kappa\geq 1$. Then, by Lemma \ref{constantbounds}, $|g_\alpha(\kappa)|\lesssim \kappa^{-1} e^{2\kappa}$. Moreover, by \eqref{KI-infty}, $K_\alpha(\kappa\rho)^2 \lesssim (\kappa\rho)^{-1} e^{-2\kappa\rho}$ for all $\rho\geq 1$. Thus,
	$$
	(\kappa \rho)^{2\alpha} | g_\alpha(\kappa)| K^2_\alpha(\kappa \rho)
	\lesssim (\kappa \rho)^{2\alpha-1} \kappa^{-1} e^{2\kappa} e^{-2\kappa\rho} \,.
	$$
	The function $z\mapsto z^{2\alpha-1} e^{-z}$ is decreasing on $((2\alpha-1)_\pp,\infty)$. Since $\kappa\rho\geq 1>(2\alpha-1)_\pp$, we deduce that
	$$
	 (\kappa \rho)^{2\alpha-1} \kappa^{-1} e^{2\kappa} e^{-2\kappa\rho} \leq  \kappa^{2\alpha-2} \leq 1 \,.
	$$
	This proves \eqref{sup-K} for $\kappa\geq 1$.
	
	Finally, let $r\leq 1 \leq r'$. It follows from \eqref{KI-zero} and \eqref{KI-infty} that
	\begin{align}
		\sup_{z>0}  \max\{z^\alpha,z^2\}   I_0(z) \big | K_\alpha(z) -K_0(z)\big|   <\infty 
		\label{KI-unif}
	\end{align}
	and
	\begin{equation}
		\label{eq:KImin}
		\sup_{z>0} \max\{ z^\alpha,z\} I_0(z) K_\alpha(z) <\infty \,.
	\end{equation}
	We use the monotonicity of $I_0$, together with \eqref{KI-unif} and \eqref{eq:KImin}, to bound
	\begin{align*}
		|\s(r,r';\kappa)| & = \ I_0(\kappa r)   \big|A^{-1}_\alpha(\kappa)  K_\alpha(\kappa r') - K_0(\kappa r') \big| \\[3pt]
		& \leq \ I_0(\kappa r')   \big|A^{-1}_\alpha(\kappa)  K_\alpha(\kappa r') - K_0(\kappa r') \big| \\[3pt]
		& \leq\   I_0(\kappa r')  |K_\alpha(\kappa r') - K_0(\kappa r') | +   |1-A^{-1}_\alpha(\kappa)| \, I_0(\kappa r')  K_\alpha(\kappa r')  \\[3pt]
		& \lesssim \min\{ (\kappa r')^{-\alpha}, (\kappa r')^{-2} \}
		+ \min\{(\kappa r')^{-\alpha},(\kappa r')^{-1}\} \big[ \kappa^{-\alpha} \1_{(0,1)}(\kappa) +\kappa^{-1} \1_{[1,\infty)}(\kappa) \big] \,. 
	\end{align*}
	We claim that the right side is bounded by $\kappa^{-2\alpha} (r')^{-\alpha}$. Since $(r')^{-\alpha}\lesssim (1+r')^{-\alpha} (1+r)^{-\alpha}$ for $r\leq 1\leq r'$, this implies the claimed bound. 
	
	For the first term on the right side, we bound $\min\{ (\kappa r')^{-\alpha}, (\kappa r')^{-2} \}\leq (\kappa r')^{-2\alpha}\leq \kappa^{-2\alpha} (r')^{-\alpha}$, as desired. For the second term we distinguish according to the size of $\kappa$. For $\kappa<1$ we bound
	$$
	\min\{(\kappa r')^{-\alpha},(\kappa r')^{-1}\} \, \kappa^{-\alpha} \leq \kappa^{-2\alpha} (r')^{-\alpha} \,,
	$$
	and for $\kappa\geq 1$ we bound, using $\kappa r'\geq 1$,
	$$
	\min\{(\kappa r')^{-\alpha},(\kappa r')^{-1}\} \, \kappa^{-1} \leq (\kappa r')^{-1} \kappa^{-1} \leq (\kappa r')^{-\alpha} \kappa^{-\alpha} = \kappa^{-2\alpha} (r')^{-\alpha} \,.
	$$
	This completes the proof of the bound that we claimed at the beginning of this step.	
	
	\bigskip
	
	\emph{Step 2.} In this step we show that $|\rr(r,r';\kappa)|\lesssim (\kappa r')^{-2\alpha}$ for all $1\leq r\leq r'$. Since $(r')^{-2\alpha} \leq (1+r)^{-\alpha} \, (1+r')^{-\alpha}$ for all such $r,r'$, we obtain the claimed bound for $\rr(r,r';\kappa)$.
	
	To prove this claim, we decompose $\rr(r,r';\kappa)$ with $1\leq r\leq r'$ further as 
	\begin{equation*} %\label{R-split-1}
		\rr(r,r';\kappa) = \rr_{\rm a}(r,r';\kappa) - \rr_{\rm b}(r,r';\kappa)
	\end{equation*}
	with 
	\begin{align*}  %\label{Ra-def}
		\rr_{\rm a}(r,r';\kappa) & :=  I_0(\kappa r) \big( K_\alpha(\kappa r') -K_0(\kappa r')\big)   
	\end{align*}
	and
	\begin{align*}  %\label{Rb-def}
		\rr_{\rm b}(r,r';\kappa) & :=
		K_\alpha(\kappa r') \big( I_0(\kappa r) -I_\alpha(\kappa r)\big).
	\end{align*}
	By \eqref{KI-unif} (together with $\max\{z^\alpha,z^2\}\geq z^{2\alpha}$) and monotonicity of $I_0$ it follows that, if $1\leq r\leq r'$, then
	\begin{align*} 
		|\rr_{\rm a}(r,r';\kappa)| & = I_0(\kappa r) \big |K_\alpha(\kappa r') -K_0(\kappa r')\big| \, \leq I_0(\kappa r') \big |K_\alpha(\kappa r') -K_0(\kappa r')\big| \,  \lesssim (\kappa r')^{-2\alpha} \,,
	\end{align*}
	which is the claimed bound.
	
	To estimate $\rr_{\rm b}(r,r';\kappa) $ we will distinguish two cases. 
	
	\smallskip
	
	\noindent If $\kappa r \geq 1$, then we use the upper bound (see \eqref{KI-infty})
	\begin{equation}\label{eq:KI-large}
		K_\alpha(z') \big | I_0(z) -I_\alpha(z)\big | \lesssim\, \frac{e^{z-z'}}{ \sqrt{z'} \, z^{3/2}}  \qquad \forall\, z, z'\geq 1 \,,
	\end{equation}
	which implies that for $z' \geq 1$ we have 
	\begin{equation*}
		\sup_{1\leq z \leq z'}  K_\alpha(z') \big | I_0(z) -I_\alpha(z)\big| \lesssim\, (z')^{-2} \leq (z')^{-2\alpha} \,,
	\end{equation*}
	which is the claimed bound.
	
	\smallskip
	
	\noindent If $\kappa r \leq 1$, then we use the bound
	$$
	\sup_{z'>0} (z')^{2\alpha} K_\alpha(z')<\infty \,,
	$$
	which follows from \eqref{KI-zero} and \eqref{KI-infty}, as well as the bound
	\begin{equation} \label{I-small}
		\sup_{0< z \leq 1} \big| I_0(z) -I_\alpha(z)\big| <\infty \,,
	\end{equation}
	which follows from \eqref{KI-zero}. Combining these two inequalities yields the claimed bound.
\end{proof}

The previous proof relies heavily on the fact that the leading terms in the asymptotic expansion of $K_\nu(z)$ for $z\to \infty$ coincide for different $\nu$; see \eqref{KI-infty}. This is used in \eqref{KI-unif} and \eqref{eq:KI-large}. This shows that subtracting the integral kernel of $(T_0+\kappa^2)^{-1}$ from $(T_\alpha+\kappa^2)^{-1}$ leads to certain cancelations in the integral kernel of $\Gamma_\alpha(\kappa)$ in terms of $\kappa^{-2\alpha} \sqrt{rr'}\, (1+r)^{-\alpha} (1+r')^{-\alpha}$. Notice in particular that the bound
$$
(T_\alpha +\kappa^2)^{-1}(r,r') \lesssim \kappa^{-2\alpha} \sqrt{rr'}\, (1+r)^{-\alpha} (1+r')^{-\alpha}
$$
{\it cannot} hold. Indeed, if this bound held, we could follow the reasoning in the proof of Proposition \ref{prop-E-1} and prove the inequality
$$
\left( \inf\spec \left( \mathfrak h^\m + v \right) \right)_\m^\alpha 
\lesssim \int_0^\infty v(r)_\m \,  (1+r)^{-2\alpha}\, r\, dr \,.
$$ 
This, however, would contradict the results of Appendix \ref{sec-one-term}.

%%%%%%%%%%%%%%%%%%%%%%%%%%%%%%%%%%%%
%%%%%%%%%%%%%%%%%%%%%%%%%%%%%%%%%%%%

\appendix

\section{On Assumption \ref{ass}}

Let us show that the parameter $\alpha$ defined in \eqref{eq:alphaflux} and the parameter $\alpha$ appearing in Assumption \ref{ass} coincide when $\mu$ is absolutely continuous with integrable density $B=\frac{d\mu}{dx}$. This follows from the following more general result.

\begin{lemma}\label{meaningalpha}
	Let $\mu$ be a signed real Borel measure on $\R^2$ of finite total variation and assume that there is an $h\in W^{1,1}_{\loc}(\R^2)$ such that $\Delta h =\mu$ in the sense of distributions. Assume that there is an $\alpha\in\R$ such that both numbers $m^\ppm$ in \eqref{h-bounds} are finite. Then $\alpha = (2\pi)^{-1} \mu(\R^2)$.
\end{lemma}

\begin{proof}
	For two parameters $\rho>0$ and $\sigma>1$ to be specified later, we introduce the function $\chi$ on $\R^2$ by
	$$
	\chi(x) := \begin{cases}
		1 & \text{if}\ |x|\leq \rho \,,\\
		1- (\ln\sigma)^{-1} \ln(|x|/\rho) & \text{if}\ \rho<|x|\leq \sigma\rho \,,\\
		0 & \text{if}\ |x|>\sigma\rho \,.
	\end{cases}
	$$
	This function is Lipschitz and has compact support. Since $h\in W^{1,1}_{\rm loc}(\R^2)$, we can test the equation $\Delta h =\mu$ against $\chi$ and obtain
	$$
	\int_{\R^2} \chi(x) d\mu(x) = - \int_{\R^2} \nabla\chi(x)\cdot\nabla h(x) \,dx \,.
	$$
	(Strictly speaking, we test the equation against the convolution of $\chi$ with a $C^\infty_c$ mollifier and pass to the limit in the mollification parameter.) Since $\mu$ has finite total variation, dominated convergence implies that the left side tends to $\mu(\R^2)$ as $\rho\to\infty$ for any fixed $\sigma$. Since $\chi$ is harmonic in $\{ \rho<|x|<\sigma\rho\}$, we see that the right side is equal to
	\begin{align*}
		- \int_{\R^2} \nabla\chi\cdot\nabla h\,dx & = - \int_{\rho<|x|<\sigma\rho} \nabla\chi\cdot\nabla h\,dx = - \int_{|x|=\sigma\rho} \frac{x}{|x|}\cdot(\nabla\chi) h\,ds(x) + \int_{|x|=\rho} \frac{x}{|x|}\cdot(\nabla\chi) h\,ds(x) \\
		& = (\ln\sigma)^{-1} (\rho\sigma)^{-1} \int_{|x|=\sigma\rho} h\,ds(x) - (\ln\sigma)^{-1} \rho^{-1} \int_{|x|=\rho} h\,ds(x) \,.
	\end{align*}
	Here $ds$ denotes integration with respect to the surface measure. Assumption \ref{ass} means that for all $x\in\R^2$,
	$$
	- \ln m^\m \leq h(x) - \alpha \ln(1+|x|/R) \leq \ln m^\pp \,.
	$$
	This implies that for $r\in\{\rho,\sigma\rho\}$, we have
	$$
	- 2\pi r \ln m^\m \leq \int_{|x|=r} h\,ds(x) - 2\pi\alpha r \ln(1+r/R) \leq 2\pi r \ln m^\pp
	$$
	and therefore
	$$
	\left| - \int_{\R^2} \nabla\chi\cdot\nabla h\,dx - (\ln\sigma)^{-1} 2\pi\alpha \ln\frac{1+\sigma\rho/R}{1+\rho/R} \right| \leq (\ln\sigma)^{-1} 2\pi \ln(m^\pp m^\m) \,. 
	$$ 
	Letting first $\rho\to\infty$ and then $\sigma\to\infty$ we easily find that
	$$
	- \int_{\R^2} \nabla\chi\cdot\nabla h\,dx \to 2\pi\alpha \,.
	$$
	This proves the claimed identity $\mu(\R^2)=2\pi\alpha$.	
\end{proof}

Next, we show that Assumption \ref{ass} is satisfied for a large class of absolutely continuous measures.

\begin{lemma}\label{hlp}
	Let $B\in\textup{\Lp}^1(\R^2,(1+\ln_\pp|x|)dx)$ such that, for some $p>1$,
	$$
	\sup_{x\in\R^2} \int_{|y-x|<1} |B(y)|^p\,dy <\infty \,.
	$$
	Then
	$$
	h(x) := \frac1{2\pi}\ \int_{\R^2} B(y) \ln |x-y|\,dy
	$$
	belongs to $W^{1,r}_{\rm loc}(\R^2)$ for $r=\frac{2p}{2-p}$ if $p<2$, any $r<\infty$ if $p=2$ and $r=\infty$ if $p>2$. It solves $\Delta h=B$ and satisfies \eqref{h-bounds} with $\alpha$ given by \eqref{eq:alphaflux}.
\end{lemma}

\begin{proof}
	The facts that $h\in W^{1,1}_{\rm loc}(\R^2)$ with
	$$
	\nabla h(x) = \frac1{2\pi} \int_{\R^2} B(y) \,\frac{x-y}{|x-y|^2}\,dy
	$$
	and that $\Delta h=B$ in the sense of distributions are in \cite[Theorem 6.21]{LiLo}. For a given $a\in\R^2$ we decompose the integral giving $\nabla h$ into the outside of $B(a,1)$ and its inside. The part from the outside is bounded in $B(a,\frac12)$ since $B\in\Lp^1(\R^2)$. The part from the inside defines a function in $\Lp^r(\R^2)$ with $r=\frac{2p}{2-p}$ if $p<2$ by the weak Young inequality \cite[Theorem 4.3]{LiLo} and the fact that $B\in\Lp^p(B(a,1))$. This also implies the assertion for $p=2$. The fact that the part from the inside is bounded for $p>2$ follows from H\"older's inequality.
	
	It remains to show that $x\mapsto h(x) - \alpha \ln(1+|x|)$ is bounded. We begin by noting that, with $\frac1p+\frac1{p'}=1$,
	$$
	\left| \int_{|y-x|<1} B(y) \ln|x-y|\,dy \right| \leq \left( \int_{|y-x|<1} |B(y)|^p\,dy \right)^\frac 1p \left( \int_{|z|<1} |\ln|z||^{p'}\,dz \right)^\frac1{p'}. 
	$$
	By assumption the right side is bounded with respect to $x$. Next, we bound
	$$
	\left| \int_{|y|<\frac12 |x|} B(y)\ln|x-y|\,dy - \ln |x| \int_{|y|<\frac12|x|} B(y)\,dy \right| \leq \int_{|y|<\frac12|x|} |B(y)| \ln \frac{|x-y|}{|x|} \,dy \,.
	$$
	For $|y|<\frac12|x|$ we have $\frac12<\frac{|x-y|}{|x|}<\frac32$, so the integral on the right side is bounded with respect to $x$. We also note that if $|x|>2$, then
	$$
	\left|  \int_{|y|<\frac12|x|} B(y)\,dy - 2\pi\alpha \right| \leq \int_{|y|\geq\frac12|x|} |B(y)|\,dy \leq \frac{1}{\ln(|x|/2)} \int_{\R^2} |B(y)|\ln_\pp|y|\,dy \,.
	$$
	This, when multiplied by $\ln|x|$ is bounded for $|x|\geq e$, say. Finally, we note that when $|x-y|\geq 1$ and $|y|\geq\frac12|x|$, then $0\leq \ln|x-y| \leq \ln(|x|+|y|)\leq\ln(\frac32|y|) = \ln\frac32 + \ln_\pp|y|$. Therefore
	$$
	\left| \int_{|y-x|\geq 1 \,,\ |y|\geq\frac12|x|} B(y)\ln|x-y|\,dy \right| \leq \int_{\R^2} (\ln\tfrac32 + \ln_\pp|y|) |B(y)|\,dy \,.
	$$
	To summarize, we have shown that $h(x) - \alpha\ln|x|$ is bounded for $|x|\geq e$. Since $\alpha\ln|x| - \alpha\ln(1+|x|)$ is bounded for $|x|\geq e$ as well, we have proved the assertion for $|x|\geq e$. The boundedness for $|x|<e$ proceeds along the same lines and we omit the details.	
\end{proof}

%%%%%%%%%%%%%%%%%%%%%%%%%%%%%%%%%%%%%%%%%%%%%%%%%%%%%%%%%%%%%%%%%%%%%%%%%%

\section{Failure of a one-term bound on $\inf\spec(\mathfrak h^-+v)$}
\label{sec-one-term}

In this section we discuss the optimality of the bound in Proposition \ref{prop-E-1}. More precisely, we shall show that for any given $0<\alpha<1$, neither
\begin{equation}
	\label{eq:apponeterm1}
	\left( \inf\spec \left( \mathfrak h^\m + v \right) \right)_\m^\alpha 
	\leq C_\alpha \, \int_0^\infty v(r)_\m^{1+\alpha} \, r \,dr
\end{equation}
nor
\begin{equation}
	\label{eq:apponeterm2}
	\left( \inf\spec \left( \mathfrak h^\m + v \right) \right)_\m^\alpha 
	\leq C_\alpha \, \int_0^\infty v(r)_\m \,  (1+r)^{-2\alpha}\, r\, dr
\end{equation}
can hold with some constant $C_\alpha$ for all real $v\in\Lp^1_{\rm loc}(\R_\pp)$. We recall that the operator $\mathfrak h^\m$ is defined in Subsection \ref{sec:redux2}.

This failure of \eqref{eq:apponeterm1} and \eqref{eq:apponeterm2} is in contrast to what happens in other related situations, for instance in \cite[Lemma 1]{KoVuWe}, where the `weak coupling term' is able to control the lowest eigenvalue uniformly. In view of this failure it is perhaps less surprising that the proof of Theorem \ref{thm-main-3} in the case $\gamma=\alpha$ is rather involved.

To show the claimed failure we apply the duality argument of Lieb and Thirring \cite{lt2} (see also \cite[Proposition 5.3]{flw-book}). This shows that \eqref{eq:apponeterm1} and \eqref{eq:apponeterm2} are equivalent to the Sobolev interpolation inequalities
\begin{align}
	\label{eq:apponeterms1}
	& \left( \int_0^\infty (1+r)^{-2\alpha} |\phi'(r)|^2\, r\,dr \right)^\frac1{1+\alpha} \left( \int_0^\infty (1+r)^{-2\alpha} |\phi(r)|^2\, r\,dr \right)^\frac\alpha{1+\alpha} \notag \\
	& \geq S_\alpha \left( \int_0^\infty (1+r)^{-2(1+\alpha)} |\phi(r)|^\frac{2(1+\alpha)}{\alpha} \,r\,dr \right)^\frac{\alpha}{1+\alpha}
\end{align}
and
\begin{equation}
	\label{eq:apponeterms2}
	\left( \int_0^\infty (1+r)^{-2\alpha} |\phi'(r)|^2\, r\,dr \right)^{1-\alpha} \left( \int_0^\infty (1+r)^{-2\alpha} |\phi(r)|^2\, r\,dr \right)^\alpha \geq S_\alpha \, \esssup_{\R_\pp} |\phi|^2 \,,
\end{equation}
respectively, with some $S_\alpha>0$ and for all locally absolutely continuous functions $\phi$ on $\R_\pp$ for which the two integrals on the left sides are finite. We shall show that neither \eqref{eq:apponeterms1} nor \eqref{eq:apponeterms2} holds.

To see that \eqref{eq:apponeterms1} fails, consider $\phi(r) = v(\epsilon r)$ with a fixed nice function $v$ and a parameter $\epsilon\ll 1$. As $\epsilon\to 0$, we compute
\begin{align*}
	\int_0^\infty |\phi'(r)|^2 \frac{r\,dr}{(1+r)^{2\alpha}} & = \epsilon^{2\alpha} \int_0^\infty |v'(s)|^2 \frac{s\,ds}{(\epsilon+s)^{2\alpha}} \sim \epsilon^{2\alpha} \int_0^\infty |v'(s)|^2 s^{1-2\alpha}\,ds \,,\\
	\int_0^\infty |\phi(r)|^2 \frac{r\,dr}{(1+r)^{2\alpha}} & = \epsilon^{-2(1-\alpha)} \int_0^\infty |v(s)|^2 \frac{s\,ds}{(\epsilon+s)^{2\alpha}} \sim \epsilon^{-2(1-\alpha)} \int_0^\infty |v(s)|^2 s^{1-2\alpha}\,ds \,,\\
	\int_0^\infty |\phi(r)|^\frac{2(1+\alpha)}{\alpha} \frac{r\,dr}{(1+r)^{2(1+\alpha)}} & = \epsilon^{2\alpha} \int_0^\infty |v(s)|^\frac{2(1+\alpha)}{\alpha} \frac{s\,ds}{(\epsilon+s)^{2(1+\alpha)}} \to \const |v(0)|^\frac{2(1+\alpha)}{\alpha} \,,
\end{align*}
where the constant on the right side of the final relation is positive. (In fact, it is equal to $\int_0^\infty (1+t)^{-2(1+\alpha)}\, t\,dt$.) Thus, the left side of \eqref{eq:apponeterms1} behaves like $\epsilon^{2\alpha^2/(1+\alpha)}$, while the right side remains positive (if $v(0)\neq 0$). Since $\alpha>0$, we arrive at a contradiction.

To see that \eqref{eq:apponeterms2} is fails, consider $\phi(r) = v(Mr)$ with a fixed, nice function $v$ and a parameter $M\gg 1$. As $M\to\infty$, we compute
\begin{align*}
	\int_0^\infty |\phi'(r)|^2 \frac{r\,dr}{(1+r)^{2\alpha}} & = \int_0^\infty |v'(s)|^2 \frac{s\,ds}{(1+s/M)^{2\alpha}} \to \int_0^\infty |v'(s)|^2 s\,ds \,, \\
	\int_0^\infty |\phi(r)|^2 \frac{r\,dr}{(1+r)^{2\alpha}} & = M^{-2} \int_0^\infty |v(s)|^2 \frac{s\,ds}{(1+s/M)^{2\alpha}} \sim M^{-2} \int_0^\infty |v(s)|^2 s\,ds \,, \\
	\esssup_{\R_\pp} |\phi|^2 & = \esssup_{\R_\pp} |v|^2 \,.
\end{align*}
Thus, the left side of \eqref{eq:apponeterms2} behaves like $M^{-2\alpha}$, while the right side is independent of $M$. Since $\alpha>0$, we arrive at a contradiction.

%%%%%%%%%%%%%%%%%%%%%%%%%%%%%%%%%%%%%%%
%%%%%%%%%%%%%%%%%%%%%%%%%%%%%%%%%%%%%%%

\bibliographystyle{amsalpha}

\end{document}